\documentclass{article}

\usepackage[T1]{fontenc}
\usepackage[utf8]{inputenc}
\usepackage{microtype}
\usepackage[rm={oldstyle,proportional}]{cfr-lm}

\usepackage[style=numeric,sorting=nyt,maxbibnames=99]{biblatex}
\DeclareCiteCommand{\plaincite}
  {\usebibmacro{prenote}}
  {\usebibmacro{citeindex}%
   \usebibmacro{cite}}
  {\multicitedelim}
  {\usebibmacro{postnote}}
\usepackage{wrapfig}
\usepackage{booktabs}
\usepackage{caption}
\usepackage[inline]{enumitem}
\setlist{itemsep=0pt,parsep=0pt}
\setlist[itemize]{label={\small\textbullet}}

\usepackage{mathtools}
\usepackage{amsthm}
\usepackage{algorithm}
\usepackage{algpseudocode}

\theoremstyle{definition} 
\newtheorem{theorem}{Theorem}[section]
\newtheorem{lemma}[theorem]{Lemma}
\newtheorem{proposition}[theorem]{Proposition}
\newtheorem{corollary}[theorem]{Corollary}

\DeclarePairedDelimiter{\abs}{\lvert}{\rvert}
\DeclarePairedDelimiter{\norm}{\lVert}{\rVert}

\DeclarePairedDelimiter{\ceil}{\lceil}{\rceil}

\DeclareMathOperator{\diam}{diam}

\newcommand{\Z}{\mathbf{Z}}

\newcommand{\R}{\mathbf{R}}

\newcommand{\level}{\ell}

\let\epsilon\varepsilon

\newcommand{\ie}{i.e.\ }

\begin{document}

\title{Routing on Heavy Path WSPD Spanners}
\author{Prosenjit Bose and Tyler Tuttle}
\date{\today}
\maketitle

\begin{abstract}
In this article, we present a construction of a spanner on a set of $n$ points in $\R^d$ that we call a heavy path WSPD spanner. The construction is parameterized by a constant $s > 2$ called the separation ratio. The size of the graph is $O(s^dn)$ and the spanning ratio is at most $1 + 2/s + 2/(s - 1)$. We also show that this graph has a hop spanning ratio of at most $2\lg n + 1$.

We present a memoryless local routing algorithm for heavy path WSPD spanners. The routing algorithm requires a vertex $v$ of the graph to store $O(\deg(v)\log n)$ bits of information, where $\deg(v)$ is the degree of $v$. The routing ratio is at most $1 + 4/s + 1/(s - 1)$ and at least $1 + 4/s$ in the worst case. The number of edges on the routing path is bounded by $2\lg n + 1$.

We then show that the heavy path WSPD spanner can be constructed in metric spaces of bounded doubling dimension. These metric spaces have been studied in computational geometry as a generalization of Euclidean space. We show that, in a metric space with doubling dimension $\lambda$, the heavy path WSPD spanner has size $O(s^\lambda n)$ where $s$ is the separation ratio. The spanning ratio and hop spanning ratio are the same as in the Euclidean case.

Finally, we show that the local routing algorithm works in the bounded doubling dimension case. The vertices require the same amount of storage, but the routing ratio becomes at most $1 + (2 + \frac{\tau}{\tau-1})/s + 1/(s - 1)$ in the worst case, where $\tau \ge 11$ is a constant related to the doubling dimension.
\end{abstract}

\tableofcontents

\section{Introduction}
\label{ch:introduction}

Routing in a graph refers to the problem of sending a message from one vertex to another. The vertex that sends the message is called the source and the vertex that is meant to receive the message is called the destination. At a vertex, a routing algorithm must decide where to forward the message in such a way that the message will eventually reach the destination. By repeatedly making forwarding decisions, one constructs a path from the source to the destination.

Centralized algorithms for computing paths in a graph are well-studied. For example, Dijkstra's algorithm \cite{Dijkstra1959} can compute the shortest path between any two vertices in a weighted graph. However, Dijkstra's algorithm requires knowledge of the entire graph. The problem is more challenging if we want distributed algorithms that must make forwarding decisions at a vertex based only on the message itself and information stored at the vertex.

A routing algorithm on a graph $G$ can be modelled as a function. The function takes as input the current vertex and the destination vertex, as well as information about the neighbourhood of the current vertex. It may take extra information that is stored either in the message itself or at the vertex. The information stored in the message is called the message header. The routing algorithm computes where the message should be forwarded to, and potentially modifies the header. Formal definitions for all of these concepts appear in Section~\ref{ch:preliminaries}.

The information that we store at a vertex is called a routing table. It is desirable for the routing tables to be as small as possible. At one extreme, the routing tables are large enough to store the shortest path tree at each vertex. At the other extreme, we only store the vertices directly adjacent to the current vertex.

If, for any pair of vertices in the graph $G$, the routing algorithm succeeds in finding a path between them, we say that the algorithm guarantees delivery on $G$. In addition to guaranteeing delivery and having small routing tables, there are other measures of the quality of a routing algorithm. Similar to the routing tables, it is desirable to make the message header as small as possible. In the best case, there is no need for a message header at all. Such an algorithm is called memoryless.

We may also be interested in bounding the length of the path produced by the routing algorithm. An algorithm with such a bound is called \textit{competitive}. In addition to bounding the length of the path, we can also bound the number of edges on the path. This is called the \textit{hop distance}.

These goals are naturally in conflict with each other. To achieve a low routing ratio on a large class of graphs, it may be necessary to store a lot of information in the routing tables or message header. A good routing algorithm must consider the trade-offs between these different parameters.

\subsection{Background}
\label{sec:background}

A geometric graph is a graph whose vertices are points in $\R^d$. The edges of a geometric graph are line segments joining their endpoints. The edges are weighted according to their length. A geometric routing algorithm is a routing algorithm that uses this position information. The routing table of a vertex can store the coordinates of its neighbours and use that information to make forwarding decisions.

Several geometric routing algorithms have been proposed. One of the simplest examples is greedy routing. Given the destination $q$ and the current vertex $u$, greedy routing forwards the message to the neighbour $v$ of $u$ that minimizes the distance $\abs{vq}$, where $\abs{vq}$ is the Euclidean distance between the points $v$ and $q$. Another example is compass routing, where it is the angle formed by the segments $uq$ and $uv$ that is minimized. For a survey of geometric routing algorithms, see \cite{Stojmenovic2002}.

A serious drawback of many routing algorithms such as these is that they do not guarantee delivery, even for very simple classes of graphs. The greedy and compass routing algorithms fail on certain triangulations, see for example \cite{Bose1999}. The message can get stuck in a loop and never reach the destination.

If a routing algorithm makes a forwarding decision based only on the positions of its current vertex, its neighbours, the destination, and potentially the source, then the routing algorithm is called memoryless. In other words, a memoryless routing algorithm does not make use of a message header. It can be shown that no deterministic memoryless routing algorithm can succeed on all graphs \cite{Bose2000}. In fact, the result is even stronger than that: no deterministic memoryless routing algorithm can guarantee delivery even when restricted to convex subdivisions.

If our goal is to design memoryless routing algorithms, then we must restrict the class of graphs under consideration. The unit disk graph is a popular model of an ad hoc wireless network, and so it is a natural arena for studying routing algorithms. In the unit disk graph, there is an edge between two vertices if and only if the distance between them is at most some radius $R$. Several routing algorithms have been proposed that guarantee delivery on the unit disk graph \cite{Bose2001,Kaplan2018}.

A \emph{competitive} routing algorithm will route a message along a path that is at most a constant times the length of the shortest possible path. The constant is called the routing ratio of the algorithm. The length of a path is defined to be the sum of the lengths of all the edges on the path. Instead of bounding the length of the path, we may be interested in bounding the number of edges on the path. We refer to this as the hop distance. This is equivalent to setting the weight of each edge to $1$. In this article a routing algorithm will be presented that constructs a path with both low length and a small number of edges.

We want to find classes of graphs that support competitive local routing. One potential class of graphs to consider are different geometric spanner graphs \cite{Narasimhan2007}. Let $G$ be a graph such that for any two vertices $u$ and $v$ we have $d_G(u, v) \le t\cdot\abs{uv}$ for some constant $t > 0$, where $d_G(u, v)$ is the length of the shortest path between $u$ and $v$. We call $G$ a $t$-spanner, and the smallest $t$ for which $G$ is a $t$-spanner the spanning ratio of $G$. Again, a more formal definition appears in Section~\ref{ch:preliminaries}. Note that this is a property of a graph, whereas the routing ratio is a property of a routing algorithm.

Many well-known geometric graphs happen to be $t$-spanners. For example, the Delaunay triangulation is known to be a $1.998$-spanner for any set of points in the plane \cite{Dobkin1990,Cui2011}. Variations on the Delaunay triangulation are also known to be spanners \cite{PaulChew1989}.

Other constructions have been devised specifically to build $(1 + \epsilon)$-spanners, for any $\epsilon > 0$. These constructions are parameterized by the desired spanning ratio, which can be made arbitrarily close to $1$. This usually requires more edges to be added to the graph, however.

The greedy spanner \cite{Chandra1992} is constructed by considering each pair of points, sorted by distance in ascending order. If there is not already a short path between the two points (\ie a path with length $(1 + \epsilon)$ times the distance between the two points), the edge is added.

Another $(1 + \epsilon)$-spanner construction is the Theta graph, independently introduced by Clarkson \cite{Clarkson1987} and Keil and Gutwin \cite{Keil1992}. To construct a Theta graph, the space around each point is partitioned into $k$ cones of equal angle. An edge is added between a point and its nearest neighbour in each cone, where the nearest neighbour is determined by projecting the points onto the bisector of the cone and taking the closest point.

Another construction technique for $(1 + \epsilon)$-spanners comes from the well-separated pair decomposition \cite{Callahan1995}. The well-separated pair decomposition (WSPD) is a data structure that stores information about distances between points in a concise manner. Intuitively, a WSPD is a partition of the edges of the complete graph on a point set, so that all the edges in a set of the partition are similar. By similar, we mean the edges have nearly the same length and orientation. The degree of similarity is controlled by a parameter $s > 2$, called the separation ratio.

A spanner can be constructed from a WSPD by adding a single edge from each set in the partition to the graph. No matter which edges are chosen, the resulting graph will be a spanner. In Section~\ref{ch:preliminaries} we will describe how to choose the edges in such a way that efficient competitive local routing is possible.

Local routing on some of these spanners has been studied in the past. Competitive local routing algorithms exist for the Delaunay triangulation \cite{Bonichon2017,Bonichon2018,Bose1999} and the Theta graph (for $k = 4$ and $k > 6$ cones) \cite{Bose2019,Bose2016}. No competitive local routing algorithm is known for the greedy spanner.

\subsection{Related work}
\label{sec:related-work}

This article considers the local routing problem on $(1 + \epsilon)$-spanners constructed from a WSPD. A few routing algorithms for this problem have been previously proposed \cite{Bose2017,Baharifard2018}. All achieve a $1 + O(1/s)$ routing ratio, where $s$ is the parameter used in the WSPD construction.

The first algorithms for routing on WSPD spanners appeared in Bose et al.\ \cite{Bose2017}. They propose two algorithms, a $2$-local and a $1$-local routing algorithm. A routing algorithm is $k$-local if the routing table of a vertex $u$ has information about not just the neighbours of $u$, but all vertices $v$ where there exists a path between $u$ and $v$ that contains at most $k$ edges. In this article, a local routing algorithm specifically refers to a $1$-local algorithm.

Neither of the algorithms in Bose et al.\ use a modifiable header, and both achieve a competitive routing ratio of $1 + O(1/s)$. A disadvantage of these algorithms is that they require the WSPD spanner to be constructed specifically to support local routing. They also require points to store coordinates in the form of bounding boxes, which makes them hard to generalize to higher dimensions and metric spaces.

The spanner construction and routing algorithm in this article are similar to those in Bose et al., but we do not require any coordinates to be stored in the routing tables. Our construction of a WSPD spanner also improves on theirs by achieving a bound of $2 \lg n + 1$ on the number of edges in the spanning path obtained\footnote{throughout this article, $\lg$ is used to mean $\log_2$}.

The other competitive local routing algorithm for WSPD spanners is from Baharifard et al.\ \cite{Baharifard2018}. Their algorithm requires a modifiable header, but it guarantees delivery on any spanner constructed from a WSPD.

All previous work has only considered sets of points in the plane. In this article we consider routing on WSPD spanners in any number of dimensions. We then further generalize the routing to metric spaces of bounded doubling dimension.

Table~\ref{table:routing-algorithms} summarizes the quantitative properties of the different routing algorithms, including the ones that are presented in this article. All of them route on a $(1 + O(1/s))$-spanner constructed from a WSPD with separation ratio $s$.

\begin{table}
\centering
\begin{tabular}{lccc}
\toprule
Algorithm & Table size at $v$ & Header size \\
\midrule
Bose et al. ($2$-local) & $O(s^2nB)$ & n/a \\
Bose et al. ($1$-local) & $O(\deg(v)B)$ & n/a \\
Baharifard et al. & $O(s^2\log\Delta)$ & $O(\log\Delta)$ \\
This article (Euclidean space) & $O(\deg(v) \log n)$ & n/a \\
This article (Doubling space) & $O(\deg(v) \log n)$ & n/a \\
\bottomrule
\end{tabular}
\caption{Comparison of routing algorithms for WSPD spanners. $s$ is the spanning ratio, $n$ is the number of points, $B$ is the number of bits needed to store an axis-aligned rectangle, and $\Delta$ is the ratio of the largest distance between two points in the input to the smallest distance.}
\label{table:routing-algorithms}
\end{table}

The problem of online routing has also been studied for metric spaces of bounded doubling dimension. See Section~\ref{sec:metric-spaces} for the relevant definitions. Previous work has focused on routing algorithms for graphs whose shortest path metric has bounded doubling dimension. Our work is first constructing a spanner in a metric space of bounded doubling dimension, and then routing on that graph. A result of Talwar \cite{Talwar2004} implies that our spanner will have bounded doubling dimension.

\begin{proposition}[{\plaincite[Proposition~3]{Talwar2004}}]
\label{prop:talwar}
Let $(X, d)$ and $(Y, \delta)$ be metric spaces such that the doubling dimension of $X$ is $\lambda$. Suppose there is a bijection $f : X \to Y$ satisfying $d(x_1, x_2) \le \delta(f(x_1), f(x_2)) \le td(x_1, x_2)$. Then the doubling dimension of $Y$ is at most $2\lambda \ceil{\lg(4t)}$.
\end{proposition}

\begin{corollary}
Let $G$ be a $t$-spanner constructed in a metric space with doubling dimension $\lambda$. Then the doubling dimension of the shortest path metric on $G$ is at most $4\lambda \ceil{\lg(4t)}$.
\end{corollary}

This implies that the shortest path metric on the heavy-path WSPD spanner has bounded doubling dimension, and so previous work applies to our graph. However, restricting attention to graphs that are constructed in a certain way as we have allows arguably simpler algorithms. In particular, our algorithm does not require a modifiable header.

\begin{table}
\centering
\begin{tabular}{lccc}
\toprule
Paper & Table size & Header size \\
\midrule
Talwar \cite{Talwar2004} & $O\big(\lambda(6/\epsilon\lambda)^\lambda\log^{\lambda + 2}(\Delta)\big)$ & $O(\lambda^2\log\lambda\log\Delta)$ \\
Chan et al. \cite{Chan2016} & $(\lambda/\epsilon)^{O(\lambda)}\log\Delta\log\delta$ & $O(\lambda \log^2 \Delta)$ \\
Abraham et al.\ \cite{Abraham2006} & $(1/\epsilon)^{O(\lambda)}\log^3 n$ & $2^{O(\lambda)}\log^3 n$ \\
This article (Doubling space) & $O(\deg(v) \log n)$ & n/a \\
\bottomrule
\end{tabular}
\caption{Comparison of routing algorithms for doubling spaces. The routing ratio is $1 + \epsilon$, $\lambda$ is the doubling dimension of the point set, and $\delta$ is the diameter of the point set.}
\label{table:routing-algorithms-doubling-space}
\end{table}

\section{Preliminaries}
\label{ch:preliminaries}

In this section we show how the spanner is constructed and analyze some of its basic properties. We also describe the various data structures that are needed for our local routing algorithm.

\subsection{Basic definitions}

A \emph{graph} $G = (V, E)$ consists of a set $V$ of vertices and a set $E$ of edges. An edge $e$ connects two vertices, called the endpoints of $e$. A \emph{geometric graph} has a set of points in $\R^d$ for vertices and a set of line segments connecting those points for edges. In a geometric graph, an edge is weighted by the Euclidean distance between its endpoints. The complete geometric graph on a point set $S$ has an edge between every pair of points in $S$.

Let $G = (V, E)$ be a graph, and let $H = (V, E')$ be a subgraph of $G$. We say that $H$ is a \emph{$t$-spanner} of $G$ if $d_H(u, v) \le t \cdot d_G(u, v)$ for all $u, v$ in $V$, where $d_G(u, v)$ denotes the length of the shortest path between $u$ and $v$ in $G$. The smallest constant $t$ for which $H$ is a $t$-spanner of $G$ is called the \emph{spanning ratio}. If we say that $H$ is a spanner, we mean that it is a $t$-spanner for some unspecified constant $t = O(1)$.

In the special case where $G$ is the complete geometric graph on some point set $S$, then a $t$-spanner of $G$ is also called a \emph{geometric spanner}. If $H$ is a geometric $t$-spanner for some point set $S$, then $d_H(u, v) \le t \cdot \abs{uv}$ for all points $u$ and $v$ in $S$.

The \emph{hop distance} between a pair of vertices $u$ and $v$ in a graph $G$ is the smallest number of edges on a path between $u$ and $v$. Let $H$ be a subgraph of $G$. If, for any pair $u$ and $v$ of vertices in $H$, the hop distance from $u$ to $v$ is at most $k$ times the hop distance from $u$ to $v$ in $G$, then we say that $H$ is a \emph{$k$-hop spanner} of $G$. This is equivalent to being a $k$-spanner in the special case where all edges have weight $1$. A geometric graph $H$ is a $k$-hop spanner if the hop distance between every pair of vertices is at most $k$.

A \emph{local routing algorithm} takes as input the current vertex $u$, the destination vertex $q$, and some information stored either at the vertex $u$ or in the message. The information stored at $u$ is called the routing table. The information that is stored with the message is called the header. Instead of storing the coordinates of the points directly, we will label the points of the graph and store the labels instead. Formally, a local routing algorithm on a graph $G = (V, E)$ is modelled as a function $f$ that takes as input four parameters in the form of bitstrings corresponding to: the label of the current vertex $u$, the label of the destination vertex $q$, the routing table $R(u)$ stored at $u$, and the message header $h$. The function computes the label of a neighbour of $u$ in $G$ that the message is forwarded to, and the new header. The header is initially empty. If the header is never modified (\ie it is always empty), then the routing algorithm is \emph{memoryless}.

The \emph{$k$-neighbourhood} of a vertex $u$ is the set of vertices $v$ at hop distance at most $k$ from $u$. If the routing table stored at a vertex $u$ contains information about the $k$-neighbourhood of $u$, then we say that the routing algorithm is \emph{$k$-local}. If $k = 1$ we say the algorithm is \emph{local}.

Let $G = (V, E)$ be a graph and let $p$ and $q$ be vertices of $G$. Starting from $p$, repeated application of the routing function $f$ will construct a path from $p$ to $q$ in $G$. Formally, let $p_1 = p$, and let $p_{i+1} = f(p_i, q, R(p_i), h_i)$, where $R(p_i)$ is the routing table of vertex $p_i$ and $h_i$ is the header computed during the previous forwarding decision. If $p_k = q$ for some integer $k$, then the routing algorithm successfully routes from $p$ to $q$. If it successfully routes for every pair of vertices in $G$, then we say that it \emph{guarantees delivery} on $G$.

A routing algorithm is \emph{competitive} on a graph if $d_R(p, q) \le t \cdot d_G(p, q)$ for all $p$ and $q$ in $G$, where $d_R(p, q)$ is the length of the path from $p$ to $q$ found by the routing algorithm. If $G$ is a geometric spanner with spanning ratio $t'$, then we have $d_R(p, q) \le tt'\abs{pq}$, so the algorithm is competitive with respect to the Euclidean distance between $p$ and $q$, not just the distance in $G$. In this article, the routing ratio of an algorithm is defined to be the smallest constant $t$ with $d_R(p, q) \le t\abs{pq}$ for all $p, q$ in $V$. With this definition, note that the routing ratio is an upper bound on the spanning ratio.

\subsection{Compressed quadtrees}
\label{sec:quadtrees}

A quadtree is a tree data structure for storing spatial data. The tree is constructed by recursively subdividing space into smaller regions. The leaves of the tree represent points, and the internal nodes represent regions of space.

Let $S$ be a set of $n$ points in $\R^d$. If $n = 1$, then the quadtree for $S$ is a single node that stores the lone point of $S$. If $n > 1$, to construct a quadtree for $S$ we need a hypercube that contains $S$. Let $C$ be a hypercube that contains $S$. We can assume that this is given to us, but if not it is simple to construct such a hypercube in time $O(dn)$.

Subdivide $C$ into $2^d$ smaller hypercubes $C_i, \dots C_{2^d}$ by bisecting it along each dimension. For each $C_i$ that contains at least two points, recursively construct a quadtree on the points in $C_i$. The root of the quadtree stores the hypercube $C$. Each of the recursively constructed quadtrees is a child of the root.

The construction as described yields a tree with $n$ leaves. Each internal node has at least one child, and at most $2^d$ children. The fact that a node can have only a single child means the height of the tree can be unbounded.

\begin{theorem}[{\plaincite[Lemma~14.1]{DeBerg2008}}]
Let $S$ be a set of points in $\R^d$. The height of a quadtree for $S$ is at most $\lg(s/m) + \frac{1}{2}\lg d + 1$, where $s$ is the side length of the initial hypercube used to construct the tree and $m$ is the minimum distance between two points of $S$.
\end{theorem}

This drawback means that the time needed to construct a quadtree can be arbitrarily large. Fortunately, there is a solution. If a quadtree has a long chain of internal nodes with only one child, then compress them all into a single edge. The resulting structure is called a compressed quadtree, and the height is now linear with respect to the number of points in the worst case.

In a (uncompressed) quadtree, each node $a$ is associated with a hypercube $C(a)$. If $a$ is at level $i$ in the tree, then the side length of $C(a)$ is $2^{-i}L$, where $L$ is the side length of hypercube associated to the root.

In a compressed quadtree, a node no longer corresponds to just one hypercube. Instead, each node $a$ corresponds to two hypercubes. A node in a compressed quadtree might correspond to an entire path in the uncompressed quadtree. We store the hypercube $C_L(a)$ that corresponds to the shallowest node on that path, and the hypercube $C_S(a)$ that corresponds to the deepest node on that path. If $p(a)$ denotes the parent of $a$, then $C_L(a)$ is obtained by splitting $C_S(p(a))$ along each dimension. Let $S(a)$ denote the set of points stored in the leaves of the subtree rooted at $a$. For any compressed quadtree node, we have $S(a) \subset C_S(a) \subseteq C_L(a)$. The two hypercubes $C_S(a)$ and $C_L(a)$ can be equal if the node $a$ does not correspond to a compressed chain of nodes in the quadtree.

Instead of constructing a quadtree (which could take an unbounded amount of time) and then compressing it afterwards, we can directly construct a compressed quadtree. There are algorithms for constructing these trees in $O(dn \log n)$ time \cite{Aluru2004}.

\begin{theorem}[{\plaincite[Section~19.2.5]{Aluru2004}}]
Let $S$ be a set of points in $\R^d$. A compressed quadtree for $S$ can be constructed in $O(dn \log n)$ time.
\end{theorem}

See \cite{Aluru2004} for a more comprehensive overview of different quadtree variants. For our application, we will need the following property of compressed quadtrees.

\begin{lemma}
\label{lem:quadtree-cell}
Let $T$ be a compressed quadtree, and let $a$ be a non-root node of $T$. The node $a$ corresponds to two hypercubes, $C_L(a)$ and $C_S(a)$. Let $\ell(a)$ be the diagonal length of $C_S(a)$. Note that this is an upper bound on the diameter of the points stored in the subtree of $a$. We have $\ell(a) \le (1/2)\ell(p(a))$, where $p(a)$ is the parent of $a$ in $T$.
\end{lemma}

\begin{proof}
By definition, the diagonal length of $C_L(a)$ is equal to $(1/2)\ell(p(a))$. Since $C_S(a) \subset C_L(a)$, the diagonal length of $C_L(a)$ is an upper bound on $\ell(a)$. Therefore $\ell(a) \le (1/2)\ell(p(a))$.
\end{proof}

\subsection{The well-separated pair decomposition}
\label{sec:wspd}

In a set of $n$ points, there are $\binom{n}{2}$ ways to select a pair of points. The well-separated pair decomposition \cite{Callahan1995} is a data structure that can approximately represent those $\Theta(n^2)$ distances in linear space. The idea behind this is that if two clusters of points are sufficiently far apart, then all of the inter-cluster distances are approximately equal. Additionally, the distance between two points in the same cluster is small relative to the distance between the clusters.

More formally, let $S$ and $T$ be two point sets in $\R^d$. We say that $S$ and $T$ are well-separated with respect to $s > 2$ if $d(S, T) \ge s \cdot \max\{\diam S, \diam T\}$, where $d(S, T) = \min\{\abs{pq}: p \in S, q \in T\}$ and $\diam S$ is the diameter of $S$, the maximum distance between two points in $S$. The number $s$ is called the separation ratio.

\begin{wrapfigure}[12]{r}{0.5\textwidth}
\centering
\includegraphics{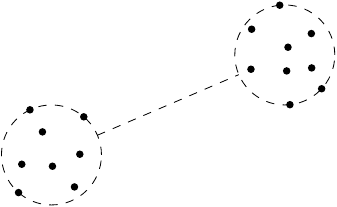}
\caption[Two point sets that are well-separated]{These two point sets are well-separated.}
\label{fig:wellseparated}
\end{wrapfigure}

There are various definitions of well-separated that appear in the literature. The definition in Narasimhan and Smid \cite{Narasimhan2007}, for example, differs from the one that we use here. The definition that we use more easily generalizes to sets of points in a metric space.

The following lemma \cite{Callahan1995} about well-separated pairs will make precise the idea that distances in one set are small compared to distances between sets.

\begin{lemma}
\label{lem:ws-pairs}
Let $S$ and $T$ be well-separated point sets with respect to $s > 2$. Then for any points $p, p' \in S$ and $q, q' \in T$
\begin{enumerate*}[label={(\alph*)},before=\unskip{: },itemjoin={{; }},itemjoin*={{; and }},after={.}]
    \item $\abs{pp'} \le (1/s)\abs{pq}$
    \item $\abs{p'q'} \le (1 + 2/s)\abs{pq}$
\end{enumerate*}
\end{lemma}

\begin{proof}
First we will prove (a). By chaining together the definitions we have
\begin{align*}
    \abs{pp'} &\le \diam S & \text{by definition of diameter} \\
    &\le (1/s) d(S, T) & \text{by definition of well-separated} \\
    &\le (1/s) \abs{pq}. & \text{by definition of $d(P, Q)$}
\end{align*}
Now we can use the triangle inequality and (a) to prove (b):
\begin{align*}
    \abs{p'q'} &\le \abs{p'p} + \abs{pq} + \abs{qq'} & \text{triangle inequality} \\
    &\le (1/s)\abs{pq} + \abs{pq} + (1/s)\abs{pq} & \text{by (a)} \\
    &= (1 + 2/s)\abs{pq}. & \qedhere
\end{align*}
\end{proof}

A well-separated pair decomposition (WSPD) of $S$ is a sequence $\{A_1, B_1\},$ $\dots, \{A_m, B_m\}$ of pairs of subsets of $S$ such that
\begin{enumerate*}[label={(\alph*)},before=\unskip{: },itemjoin={{; }},itemjoin*={{; and }},after={.}]
    \item $A_i \cap B_i = \emptyset$ for all $i$
    \item for each pair $p, q$ of points in $S$ there is exactly one $i$ such that $p \in A_i$ and $q \in B_i$ (or $p \in B_i$ and $q \in A_i$)
    \item $A_i$ and $B_i$ are $s$-well separated for all $i$
\end{enumerate*}

Given a compressed quadtree $T$, we can construct a WSPD with a recursive algorithm \cite[Section~3.1.1]{Har-peled2008}. Let $S$ be a set of points and let $T$ be a compressed quadtree for $S$. For an internal node $a$ of $T$, let $C_S(a)$ be the (smaller) hypercube represented by $a$, let $\ell(a)$ be the diagonal length of $C_S(a)$, and let $S(a)$ be the subset of $S$ in $C_S(a)$. Calling the following procedure with the root of $T$ as both arguments, i.e.\ the initial call is $\Call{WSPD}{r,r}$ where $r$ is the root of $T$, will result in a WSPD.

\begin{algorithm}
\caption{Construction of a WSPD}
\label{alg:wspd}
\begin{algorithmic}
\State \textbf{Input}: $a$ and $b$ are nodes of a compressed quadtree $T$ that stores a set $S$ of points in its leaves, $s > 2$ is the separation ratio
\State \textbf{Output}: if initially called with both $a$ and $b$ equal to the root of $T$, the algorithm outputs a WSPD of $S$ with separation ratio $s$
\Procedure{WSPD}{$a, b$}
\State \algorithmicif\ $a = b = \{p\}$ for some point $p$\ \algorithmicthen\ \Return $\emptyset$
\If{$\ell(a) < \ell(b)$}
    \State swap $a$ and $b$
\EndIf
\Comment{now $\ell(a) \ge \ell(b)$}
\If{$d(C_S(a), C_S(b)) \ge s \cdot \max\{\ell(a), \ell(b)\}$}
    \State \Return $\{\{S(a), S(b)\}\}$
    \Comment{$C_S(a)$ and $C_S(b)$ are well-separated}
\Else
    \State let $a_1, a_2, \dots, a_k$ be the children of $a$
    \State \Return $\bigcup_{i=1}^k \Call{WSPD}{a_i, b}$
\EndIf
\EndProcedure
\end{algorithmic}
\end{algorithm}

If we assume that the nodes of $T$ store their diagonal length and containing hypercube, then each call to this algorithm takes $O(d)$ time, excluding the work done in the recursive calls. Computing the distance between two hypercubes takes $O(d)$ time, and all other operations take constant time.

By considering the tree of recursive calls made, we can show that there are $O(m)$ calls in total, where $m$ is the number of pairs returned by the algorithm. The leaves of this recursion tree correspond to calls that returned a well-separated pair, and the internal nodes correspond to calls that recursed. The leaves represent the well-separated pairs in the WSPD, so there are exactly $m$ leaves.

Since every internal node in the compressed quadtree $T$ has at least two children, the internal nodes of the recursion tree also all have at least two children. The number of internal nodes cannot exceed $m$. Therefore since there are $O(m)$ nodes in the recursion tree the running time of this algorithm is $O(dm)$. What remains to be shown is that $m$ is linear in $n = \abs{S}$, the number of points.

\begin{theorem}[{\plaincite[Lemma~3.9]{Har-peled2008}}]
Let $T$ be a compressed quadtree for a set $S$ of $n$ points in $\R^d$. The number of pairs returned by Algorithm~\ref{alg:wspd}, with the root of $T$ as both arguments, is $O(s^dn)$. The running time of the algorithm is $O(ds^dn)$.
\end{theorem}

The WSPD that results from this algorithm has an important property that we will need in Section~\ref{sec:hpw-spanners}, so we will prove it now.

\begin{lemma}
\label{lem:wspd-subtree}
Let $T$ be a compressed quadtree for some point set $P$, and let $W$ be a WSPD computed with Algorithm~\ref{alg:wspd}. Every pair in $W$ has the form $\{S(a), S(b)\}$ for some nodes $a, b$ of $T$. Let $p, q$ be any two points of $P$ and let $\{S(a), S(b)\}$ be the pair that separates them. If $c$ is a node that stores both $p$ and $q$ in its subtree, then $a$ and $b$ are both descendants of $c$.
\end{lemma}

\begin{proof}
Assume that $p \in S(a)$ and $q \in S(b)$. The sets $S(a)$ and $S(b)$ are disjoint \cite{Har-Peled2006}. Therefore, one of $a$ or $b$ cannot be an ancestor of the other. Also, $a$ must lie on the path from the leaf storing $p$ to the root of $T$. Likewise for $b$.

Let $d$ be the least common ancestor of $p$ and $q$. Since $S(a)$ and $S(b)$ are disjoint, both $a$ and $b$ must be descendants of $d$. Since $d$ is defined to be the deepest node that stores both $p$ and $q$ in its subtree, $c$ must be $d$ or an ancestor of $d$. See Figure~\ref{fig:wspd-subtree}.
\end{proof}

\begin{figure}
\centering
\includegraphics{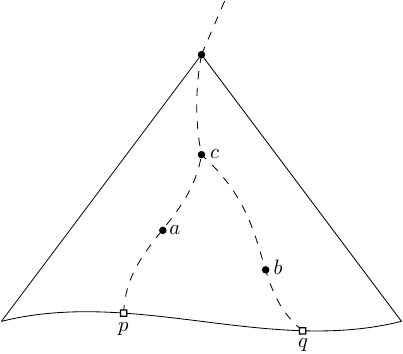}
\caption[Illustration of Lemma~\ref{lem:wspd-subtree}]{Illustration of Lemma~\ref{lem:wspd-subtree}. The points $p$ and $q$ are separated by the pair $\{S(a), S(b)\}$, where $a$ and $b$ are nodes in the compressed quadtree. The node $c$ contains both $p$ and $q$ in its subtree.}
\label{fig:wspd-subtree}
\end{figure}

Finally, the construction of a WSPD from a point set can be summarized in the following theorem.

\begin{theorem}[{\plaincite[Theorem~3.10]{Har-peled2008}}]
Given a set $S$ of $n$ points in $\R^d$, a WSPD with separation ratio $s > 2$ with $O(s^dn)$ pairs can be constructed in time $O(d(n \log n + s^dn))$.
\end{theorem}

In the construction of a WSPD, we used compressed quadtrees. There are alternative tree structures that have been used instead. For example, the fair split tree of Callahan and Kosaraju \cite{Callahan1995}. A WSPD with a linear number of pairs can be constructed from the fair split tree, with the same time bound. The fair split tree is a binary tree, which can be desirable for some applications.

The reason compressed quadtrees are used is so that we can use Lemma~\ref{lem:quadtree-cell}. The diameter of the hypercube representing a compressed quadtree node is a constant fraction of the diameter of its parent. In a fair split tree this fraction depends on the dimension of the point set since a split is only done along one dimension, instead of along all dimensions simultaneously. This results in having to go $d$ levels down the fair split tree before the diameter of a node is a constant fraction of the diameter of its ancestor. This property of compressed quadtrees will be used in the analysis of the routing algorithm.

\subsection{The heavy path decomposition}
\label{sec:heavy-path}

We now describe the heavy path decomposition of a tree. Let $T$ be a rooted tree. If $a$ is a node of $T$, then the size of $a$ is the number of leaves in the subtree rooted at $a$. It is worth noting here that a node $a$ is considered to be an ancestor and a descendant of itself. For each internal node $a$, choose one child whose subtree size is maximal among all children of $a$ (breaking ties arbitrarily), and mark the edge from $a$ to that child as heavy. The other edges are marked as light. If $b$ is a child of $a$ and the edge from $a$ to $b$ is heavy, then we call $b$ the heavy child of $a$. Otherwise $b$ is a light child of $a$. What results is a decomposition of the tree into heavy paths, one for each leaf node. The heavy path decomposition of a tree with $n$ leaves can be computed in $O(n)$ time \cite{Sleator1983}.

For an internal node $a$, let $r(a)$ be the leaf node defined by following the unique heavy path down the tree starting from $a$. This leaf is called the representative of $a$. Let $h(a)$ be the node defined by following the heavy path up the tree, again starting from $a$, until the edge to the parent is no longer heavy.

\begin{lemma}
\label{lem:light-paths}
The number of light edges on any root-to-leaf path in a heavy path decomposition of a compressed quadtree is at most $\lg n$, where $n$ is the number of leaves in the compressed quadtree.
\end{lemma}

\begin{proof}
Let $T$ be a tree for which we have computed a heavy path decomposition. Let $a$ be an internal node of $T$, and let $b$ be any light child of $a$. The size of $b$ must be at most half the size of $a$, otherwise the edge to $b$ would be marked heavy. Since following a light edge reduces the size of the subtree by at least a factor of two, the number of such edges on a root-to-leaf path is at most $\lg n$.
\end{proof}

\begin{figure}
\centering
\includegraphics{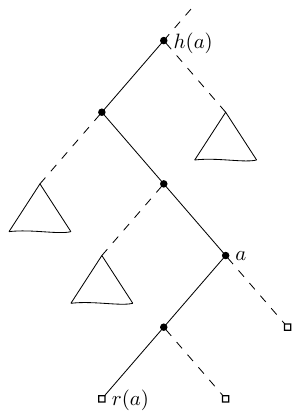}
\caption[A heavy path decomposition]{A heavy path in a tree. The leaf $r(a)$ is the representative of every node on the path from $r(a)$ to $h(a)$.}
\label{fig:ex-heavypath}
\end{figure}

\noindent
\begin{minipage}{0.65\textwidth}
\begin{lemma}
\label{lem:representative-hierarchy}
Let $T$ be a tree, and let $a$ be an internal node of $T$. Compute a heavy path decomposition of $T$. Let $r(a)$ be the representative of $a$. Then for every node $b$ on the path from $h(a)$ to $r(a)$, we have $r(b) = r(a)$.
\end{lemma}

\begin{proof}
The representative of $b$ is the leaf node defined by following the heavy path down the tree. Since $b$ is contained on the heavy path from $h(a)$ to $r(a)$, and there is a unique heavy path from each node down to a leaf, $r(a)$ must also be the representative of $b$.
\end{proof}
\end{minipage}
\hfill
\begin{minipage}{0.3\textwidth}
\includegraphics{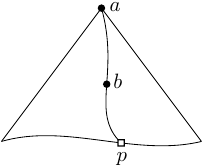}
\captionof{figure}[Illustration of Lemma~\ref{lem:representative-hierarchy}]{Illustration of Lemma~\ref{lem:representative-hierarchy}. The point $p$ is the representative of both $a$ and $b$.}
\label{fig:heavypath}
\end{minipage}

\subsection{Constructing a heavy path WSPD spanner}
\label{sec:hpw-spanners}

Constructing a spanner graph given a WSPD is simple. For each pair in the WSPD, choose an arbitrary point from each set and add an edge between those two points. The result is a $t$-spanner for $t = (s+4)/(s-4)$, where $s > 4$ is the separation ratio of the WSPD \cite{Narasimhan2007}. Since we can choose these points in any manner, we are free to decide on a scheme that benefits our application. In this article, we will choose the points using a method based on the heavy path decomposition described in the previous section. The spanner that we construct will be called a heavy path WSPD spanner.

Let $T$ be the compressed quadtree used to compute the WSPD. Compute a heavy path decomposition of $T$. For each pair $\{A, B\}$ in the WSPD, there is a corresponding pair $\{a, b\}$ of nodes in $T$. The edge that we add to the graph will be between the points $r(a)$ and $r(b)$.

Now we will prove that this graph is a $(1 + 2/s + 2/(s-1))$-spanner. To construct a path between two points $p$ and $q$, consider Algorithm~\ref{alg:spanner-path}. Let $\{S(a), S(b)\}$ be the WSPD pair that separates $p$ from $q$. The algorithm adds an edge between $r(a)$ and $r(b)$, and recursively constructs a path from $p$ to $r(a)$ and from $r(b)$ to $q$.
\begin{algorithm}
\caption{Constructing a short path in a heavy path WSPD spanner}
\label{alg:spanner-path}
\begin{algorithmic}
\State \textbf{Input}: Two points $p$ and $q$ in a heavy path WSPD spanner
\State \textbf{Output}: A path between $p$ and $q$
\Procedure{BuildPath}{$p, q$}
\If{$p = q$} \Comment{Base case}
    \State \Return $\emptyset$
\Else
    \State let $\{S(a), S(b)\}$ be the WSPD pair that separates $p$ from $q$
    \State \Return $\Call{BuildPath}{p, r(a)} \cup r(a)r(b) \cup \Call{BuildPath}{r(b), q}$
\EndIf
\EndProcedure
\end{algorithmic}
\end{algorithm}

To analyze the spanning ratio, we will first consider a special case, where $q$ is the representative of a node storing $p$. In other words, $h(q)$ is an ancestor of $h(p)$. In this case, in every call made to \textsc{BuildPath} that does not immediately return, at least one of the two recursive calls will be to the base case.

\begin{lemma}
\label{lem:span1}
Let $S$ be a set of points in $\R^d$, and let $T$ be a compressed quadtree for the points of $S$. Construct a heavy path WSPD spanner for $S$. Let $p$ and $q$ be points stored in the leaves of $T$ such that $q$ is the representative of some node containing $p$. Note that $r(p)$ is not necessarily equal to $r(q)$. In a call to \Call{BuildPath}{$p, q$}, at most one edge is added at each level of recursion.
\end{lemma}

\begin{proof}
Let $c$ be the lowest common ancestor of both $p$ and $q$ with $q = r(c)$. The proof is by induction on the size of $S(c)$. In the base case, $S(c) = 1$. That means $p = q$, and so the algorithm returns immediately, adding no edges.

Now consider the case $p \ne q$. Let $\{S(a), S(b)\}$ be the pair that separates $p$ from $q$. By Lemma~\ref{lem:wspd-subtree}, both $a$ and $b$ are descendants of $c$. Since $a$ is a descendant of $c$, we have $\abs{S(a)} < \abs{S(c)}$. By induction, at most one edge is added in each level of recursion when constructing the path from $p$ to $r(a)$.

By Lemma~\ref{lem:representative-hierarchy}, $q = r(b)$. One of the recursive calls is \Call{BuildPath}{$r(b), q$}, which will return without adding any edges. Since only the edge $r(a)r(b)$ is being added in the original call, the induction holds.
\end{proof}

Consider an initial call to \Call{BuildPath}{$p, q$}, where $q$ is not necessarily the representative of an ancestor of $p$. Two recursive calls are made, to \Call{BuildPath}{$p, r(a)$} and \Call{BuildPath}{$r(b), q$}, respectively. Both of these calls satisfy the conditions for Lemma~\ref{lem:span1}.

We can also bound the length of the edges being added to the path, as a function of the recursion depth.

\begin{lemma}
\label{lem:span2}
Let $S$ be a set of points in $\R^d$, and let $T$ be a compressed quadtree for the points of $S$. Construct a heavy path WSPD spanner for $S$. Let $p$ and $q$ be points of $S$. Consider the series of recursive calls made during a call to \Call{BuildPath}{$p,q$}. If the level of recursion of some call is $k$, the length of the edge added during that call is at most $(1/s)^k\abs{pq}$.
\end{lemma}

\begin{proof}
The proof is by induction on the depth of recursion. In the base case, $k = 1$. Let $\{S(a), S(b)\}$ be the pair separating $p$ from $q$. Assume without loss of generality that we are in the call to \Call{BuildPath}{$p, r(a)$}. Every edge added by this call and the recursive calls made by it will be in the subtree of $a$. By Lemma~\ref{lem:ws-pairs}, any such edge has length at most $(1/s)\abs{pq}$.

Now assume that it is true for recursion up to some depth $k \ge 1$. Consider a call \Call{BuildPath}{$x, y$} made at recursion depth $k$. Let $\{S(c), S(d)\}$ be the pair that separates $x$ from $y$. By induction the length of the edge $r(c)r(d)$ is at most $(1/s)^k\abs{pq}$. By Lemma~\ref{lem:ws-pairs} the length of any edge in $S(c)$ or in $S(d)$ is at most $(1/s)\abs{r(c)r(d)} \le (1/s)^{k+1}\abs{pq}$.
\end{proof}

Using these two lemmas, we can bound the spanning ratio of the path between $p$ and $q$. Note that this implies the graph is connected, since a graph with finite spanning ratio must have a path between any two vertices.

\begin{theorem}
\label{thm:spanner}
Let $S$ be a set of $n$ points in $\R^d$. The heavy path WSPD spanner $G$ for $S$ has a spanning ratio of at most $1 + 2/s + 2/(s - 1)$.
\end{theorem}

\begin{proof}
Let $p$ and $q$ be points of $S$. Algorithm~\ref{alg:spanner-path} constructs a path from $p$ to $q$. The edge added from $r(a)$ to $r(b)$ has length at most $(1 + 2/s)\abs{pq}$ by Lemma~\ref{lem:ws-pairs}. The length of the path from $p$ to $r(a)$ can be bounded using Lemma~\ref{lem:span1} and Lemma~\ref{lem:span2}. There is at most one edge being added at each level of recursion, and the length of the edge being added at level $k$ is at most $(1/s)^k\abs{pq}$. Let $M$ be the maximum recursion depth. Therefore, the length of the path from $p$ to $r(a)$ is at most
\begin{equation*}
    \sum_{k=1}^M \bigg(\frac{1}{s}\bigg)^k \abs{pq} \le \frac{1}{s - 1}\abs{pq}.
\end{equation*}
The length of the path from $r(b)$ to $q$ can be bounded in the same way. Therefore the total length of the path is at most
\begin{equation*}
    \frac{1}{s - 1}\abs{pq} + \bigg(1+\frac{2}{s}\bigg)\abs{pq} + \frac{1}{s - 1}\abs{pq} = \bigg(1 + \frac{2}{s} + \frac{2}{s-1}\bigg)\abs{pq}. \qedhere
\end{equation*}
\end{proof}

In addition to bounding the length of the path, we can bound the number of edges on the path from $p$ to $q$, using Lemma~\ref{lem:light-paths}. This is because every edge on the spanner path ``traverses'' at least one light edge. The diameter of a spanner is the maximum number of edges over all the shortest paths between any pair of points in the spanner.

\begin{lemma}
\label{lem:spanner-diameter}
For two points $p$ and $q$ in a heavy path WSPD spanner, the number of edges on the path from $p$ to $q$ found by \Call{BuildPath}{$p, q$} is at most $2\lg n + 1$. In other words, the heavy path WSPD spanner is a $(2\lg n + 1)$-hop spanner.
\end{lemma}

\begin{proof}
Let $\{S(a), S(b)\}$ be the WSPD pair that separates $p$ from $q$. Consider the subpath $p = p_1, p_2, \dots, p_k = r(a)$. For each edge $p_ip_{i+1}$, we know that $p_i$ is contained in the subtree of $h(p_{i+1})$, where $h(p_{i+1})$ is the shallowest node in the compressed quadtree that $p_{i+1}$ is the representative of.

The sequence of nodes $h(p_1), h(p_2), \dots, h(p_k)$ must then all lie on the same root-to-leaf path (that is, the path from $p$ to the root). Since all of these nodes have different representatives, by Lemma~\ref{lem:light-paths} there can be at most $\lg n$ of them.

The same is true for the subpath between $r(b)$ and $q$, and then adding the edge between $r(a)$ and $r(b)$ gives an upper bound of $2\lg n + 1$ edges on the spanner path.
\end{proof}

We end this section with a theorem that summarizes the entire construction of a heavy path WSPD spanner and all its properties. Note that the \textsc{BuildPath} algorithm that finds a path between two points in a heavy path WSPD spanner is not a local algorithm since it requires knowing the pair $\{S(a), S(b)\}$ from $p$.

\begin{theorem}
\label{thm:hpws-summary}
Let $S$ be a set of $n$ points in $\R^d$, and let $s > 2$. In $O(d(n \log n + s^dn))$ time, we can construct a graph $G$ called a heavy path WSPD spanner with the following properties:
\begin{itemize}
    \item The number of edges in $G$ is $O(s^dn)$.
    \item $G$ is a $(1 + 2/s + 2/(s - 1))$-spanner.
    \item $G$ is a $(2\lg n + 1)$-hop spanner.
\end{itemize}
Additionally, between any two points there is a single path (found by algorithm \textsc{BuildPath}) that achieves both the spanning and hop-spanning ratio.
\end{theorem}

\section{Local routing in Euclidean space}
\label{ch:routing-euclidean}

In this section, we present a local routing algorithm for heavy path WSPD spanners. Let $S$ be a set of points, $T$ be a compressed quadtree for $S$, $W$ be a WSPD computed using $T$, and $G$ be a heavy path WSPD spanner constructed as described in the previous section. The main difficulty here is trying to make a decision at a point $p$ \textit{locally}, without knowing the neighbourhood of the destination. We now have all the tools to describe the routing algorithm. First, we will explain what we need to store at each vertex. Then, we show how to use this information to design a routing algorithm and analyze it.

\subsection{Routing tables}
\label{sec:routing-tables}

First, we describe a labelling scheme for the nodes of $T$. The vertices of $G$ will store these labels. The message will only use the label of the destination to route. In other words, the algorithm is memoryless.

Each leaf will get a unique label in the range $1, 2, \dots, n$. Perform a depth-first traversal of $T$, and label the leaves in the order that they are visited. The label of an internal node will be the set of all the labels in the leaves of that node's subtree. We call this the DFS labelling scheme. The labelling scheme ensures that this set will be an interval. This fact is well-known but we include a proof for completeness.

\begin{lemma}
Let $T$ be a tree, and label its leaves using a depth first search. Let $a$ be a node of $T$. Let $I$ be the set of labels of the leaves in the subtree rooted at $a$. The labels form a contiguous subset of $\{1, 2, \dots, n\}$. That is, if $i$ is the minimum label and $j$ is the maximum label in $I$, then $I = \{i, i+1, \dots, j-1, j\}$.
\end{lemma}

\begin{proof}
We prove this by induction on the size of the tree. Let $n(a)$ be the size of the subtree rooted at a node $a$. If $n(a) = 1$ then $a$ is a leaf and the label of $a$ is a single integer.

If $n(a) > 1$, then $a$ is an internal node and it has children $a_1, \dots, a_k$. Assume that $a_1$ is the first child visited in the depth-first search, $a_2$ is the second child visited, and so on. By induction, the label of each child node is an interval. They must be disjoint since the label of each leaf is unique. Let $[x_i, y_i]$ be the label of $a_i$, for all $1 \le i \le k$, where $x_i$ is the minimum label and $y_i$ is the maximum label. Since depth-first search will visit the nodes in the subtree of $a_{i+1}$ immediately after visiting the nodes in the subtree of $a_i$, we must have $x_{i+1} = y_i + 1$ for all $1 \le i < k$. Therefore when we take the union of all $k$ intervals, we get an interval of the form $[x_1, y_k]$.
\end{proof}

Since we only need to store the minimum and maximum labels of each interval,
we only need $2\lg n$ bits to store the label of an internal node.

In a depth-first search, the children of a node can be visited in any order. If we always visit the child with the largest subtree first (\ie, always follow the heavy edge), and then visit the other children in an arbitrary order, then we can save some memory as shown in the following lemma. We call this type of DFS labelling scheme a \textit{heavy path DFS} labelling scheme.

\begin{lemma}
\label{lem:rep-label}
Let $T$ be a compressed quadtree and let $a$ be an internal node of $T$. In a heavy path DFS labelling of $T$, the label of $r(a)$ is the minimum label of all the points stored in the subtree of $a$. That is, if the label of $a$ is $[x, y]$, then the label of $r(a)$ is $x$.
\end{lemma}

\begin{proof}
The proof is by induction on the size of the tree. Let $a$ be a node of $T$. If $a$ is a leaf, then the label of $a$ is of the form $[x, x]$, where $x$ is the label of the point stored at $a$.

Now assume that $a$ is an internal node, and that $b$ is the first child visited by the heavy path DFS labelling. By definition, $a$ and $b$ are on the same heavy path, and so $r(a) = r(b)$. If $[x, y]$ is the label of $b$, then by induction the label of $r(a) = r(b)$ is $x$. Since $b$ is the first child of $a$ visited, the labels of all the points in $b$ are smaller than the labels of the points stored in the subtrees of other children of $a$. Therefore, $x$ is the smallest label in the subtree of $a$.
\end{proof}

We now describe the information that needs to be stored in the routing table of a vertex $u$ of the graph $G$. First, store the label of $u$. Second, for each neighbour $v$ of $u$, let $\{S(a_v), S(b_v)\}$ be the WSPD pair that generated the edge between $u$ and $v$, where $v \in S(b_v)$, \ie $v$ is the leaf in the subtree of $b_v$ with minimum label. Store the labels (defined by the heavy path DFS labelling of $T$) of $v$, $b_v$, and $h(v)$. Recall that $h(v)$ is the shallowest node in $T$ for which $v$ is a representative. Notice that the label of $a_v$ is not stored, as it is never used by the routing algorithm. Furthermore, since there is an edge between $u$ and $v$ we know that $u$ is the representative of $S(a_v)$.

\begin{lemma}
\label{lem:memory}
The total size of the routing tables is $O(s^dn\log n)$.
\end{lemma}

\begin{proof}
The label of a point is a single integer in the range $\{1, \dots, n\}$, and the label of an internal node is two integers in the same range, so in total we need to store $5\lg n$ bits for each neighbour of $u$: we need $\lg n$ bits to store the label of $v$, and $2\lg n$ bits each for the labels of $b_v$ and $h(v)$. This can be improved by applying Lemma~\ref{lem:rep-label}. We know that $v$ is the representative of $b_v$, since the edge $uv$ was generated by the pair $\{a_v, b_v\}$. We also know that $v$ is the representative of $h(v)$, by the definition of $h(v)$. So if $x$ is the label of $v$, then the labels of $b_v$ and $h(v)$ are of the form $[x, y]$ and $[x, z]$, respectively. Since three of the integers being stored are equal, we actually only need $3\lg n$ bits.

The total size of the routing table at a vertex $u$ of $G$ is $(3\deg(u) + 1)\lg n$, where $\deg(u)$ is the number of neighbours of $u$ in $G$. The total size of the routing tables stored in the entire graph is
\begin{equation*}
    \sum_{u \in P} (3\deg(u) + 1)\lg n = (6m + n)\lg n
\end{equation*}
where $m$ is the number of edges in the spanner. Since we know $m = O(s^dn)$, the total size of the routing tables is therefore $O(s^dn\log n)$.
\end{proof}

\subsection{Routing in a heavy path WSPD spanner}
\label{sec:routing-algorithm}

We now present the routing algorithm. Let $p$ be the starting vertex, and let $q$ be the destination vertex. We can assume that the label of the destination $q$ is stored with the message. No other information is stored with the message (that is, the algorithm is memoryless). The algorithm proceeds in two stages: the ascending stage and the descending stage. We first check if we are in the descending stage of the algorithm. If so, perform a descending step. If not, perform an ascending step. We will refer to this algorithm as the heavy path routing algorithm. Let $u$ be the current vertex.

\begin{enumerate}
    \item{} [\textsl{Descending step}] If $u$ has a neighbour $v$ (with WSPD pair $\{S(a_v), S(b_v)\}$) such that $q \in b_v$, then forward the message to $v$.
    \item{} [\textsl{Ascending step}] Otherwise, find the representative of the parent of $h(u)$, and forward the message to that vertex.
\end{enumerate}

The proof that this routing algorithm guarantees delivery is split into two stages. Let $\{S(a), S(b)\}$ be the WSPD pair separating $p$ from $q$. First we will prove that the ascending step will be applied until the message reaches $r(a)$. Then, we will prove that the descending step will be applied until the message reaches its destination. That is, the routing algorithm can be split into two ``stages'': a series of ascending steps followed by a series of descending steps.

First we need to show that it is possible to implement an ascending step using only the information stored in the vertices of $G$.

\begin{lemma}
\label{lem:ascending1}
The representative of the parent of $h(u)$ is a neighbour of $u$ in $G$, and can be found using only the information in the routing table at $u$.
\end{lemma}

\begin{proof}
Let $z$ be the parent of $h(u)$, and let $v$ be the point that represents $z$, \ie $v = r(z)$. Consider the WSPD pair $\{S(c), S(d)\}$ that separates $u \in S(c)$ from $v \in S(d)$. Both $c$ and $d$ must be descendants of $z$, since $u$ and $v$ are both in the subtree rooted at $z$.

We know that $c$ is somewhere on the path from $u$ to $h(u)$, and since $u$ is the representative of every node on that path, it is the representative of $c$. Likewise, since $v$ is the representative of every node on the path from $v$ to $a$, it must be the representative of $d$. Therefore, there must be an edge between $u$ and $v$.

\begin{figure}
\centering
\includegraphics{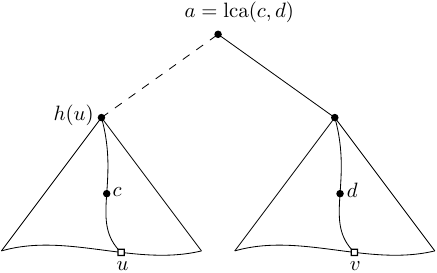}
\caption[Illustration of Lemma~\ref{lem:ascending1}]{Illustration of Lemma~\ref{lem:ascending1}. There must be an edge between $u$ and $v$ in the heavy path WSPD spanner.}
\label{fig:ascending1}
\end{figure}

To find that edge using the routing tables, iterate over all neighbours of $u$, finding the neighbour $w$ such that $u \in h(w)$ and $h(w)$ is as small as possible. That is, $h(w)$ is a descendant of $h(x)$ for any other neighbour $x$ with $u \in h(x)$. The point $w$ is the representative of the parent of $h(u)$.
\end{proof}

Another way of viewing Lemma~\ref{lem:ascending1} is that one application of the ascending step will move the message one light edge ``up'' in the quadtree.

The next two lemmas prove that, from $p$, the routing algorithm will repeatedly apply an ascending step until $r(a)$ is reached. This part of the algorithm is called the ascending stage. Let $u$ refer to the current vertex.

\begin{lemma}
\label{lem:ascending2}
Starting from $p$, repeated application of the ascending step will forward the message to $r(a)$.
\end{lemma}

\begin{proof}
The node $a$ is an ancestor of $p$ in the quadtree. By Lemma~\ref{lem:ascending1}, repeatedly applying the ascending step will send the message to the representative of every node on the path from $p$ to the root in turn. Since $r(a)$ is one of these representatives, eventually the message reaches $r(a)$.
\end{proof}

\begin{lemma}
\label{lem:ascending3}
The ascending step is always applied if $u$ is a leaf in the subtree rooted at $a$, but not equal to $r(a)$.
\end{lemma}

\begin{proof}
Consider the pair $\{S(c), S(d)\}$ separating $u$ from $q$. Assume that $r(c) = u$, so that a descending step is applied. Recall that $\{S(a), S(b)\}$ is the WSPD pair that separates $p$ from $q$. Since there can only be one pair in the WSPD separating $u$ from $q$, we must have $\{S(c), S(d)\} = \{S(a), S(b)\}$, and so $u = r(a)$. Therefore the only way for the descending step to be applied if $u$ is in the subtree of $a$ is for $u$ to be equal to $r(a)$.
\end{proof}

Once the message reaches $r(a)$, the descending step will be applied until the destination is reached. In fact, the path constructed by the descending steps from $r(a)$ to $q$ is identical to the path from $r(a)$ to $q$ constructed by Algorithm~\ref{alg:spanner-path}, as the following lemma shows.

\begin{lemma}
\label{lem:descending}
If $u$ is on the path constructed in Algorithm~\ref{alg:spanner-path} from $p$ to $q$ and $u \not\in S(a)$, then the descending step is applied and $u$ forwards the message to the next point on the spanner path constructed by Algorithm~\ref{alg:spanner-path}.
\end{lemma}

\begin{proof}
Consider the construction of the spanner path from Algorithm~\ref{alg:spanner-path}. Let $\{c, d\}$ be the WSPD pair separating $u$ from $q$. An edge from $r(c)$ to $r(d)$ is added to the path. In both the spanner construction and the descending step, $u$ is the representative of $c$.
\end{proof}

Putting the ascending and descending steps together will therefore successfully route a message from $p$ to $q$.

\begin{theorem}
\label{thm:routing-correct}
The heavy path routing algorithm will successfully route a message in a heavy path WSPD spanner, with information stored in each vertex as outlined in Section~\ref{sec:routing-tables}.
\end{theorem}

\begin{proof}
Starting at $p$, the message will reach $r(a)$ by repeatedly applying the ascending step by Lemmas~\ref{lem:ascending2} and \ref{lem:ascending3}. Then Lemma~\ref{lem:descending} implies that the message will be forwarded along the path of Algorithm~\ref{alg:spanner-path} from $r(a)$ to $q$.
\end{proof}

\subsection{Analysis of the local routing algorithm}
\label{sec:routing-analysis}

In this section we will bound the routing ratio of the heavy path local routing algorithm. First we will bound the length of the path found in the descending stage, as it is much easier to do.

\begin{lemma}
\label{lem:descending-analysis}
Let $p$ and $q$ be points in a heavy path WSPD spanner. The length of the path constructed during the descending stage of the heavy path local routing algorithm is at most $(1 + 2/s + 1/(s-1))\abs{pq}$.
\end{lemma}

\begin{proof}
Lemma~\ref{lem:descending} implies that the descending stage finds a path from $r(a)$ to $q$, where $r(a)$ is the representative of the set containing $p$ in the WSPD pair $\{S(a), S(b)\}$ that separates $p$ from $q$.

From Theorem~\ref{thm:spanner}, we know that the edge $r(a)r(b)$ has length at most $(1 + 2/s)\abs{pq}$, and that the subpath constructed in Algorithm~\ref{alg:spanner-path} from $r(b)$ to $q$ has length at most $1/(s-1)\abs{pq}$. Lemma~\ref{lem:descending} implies that the path from $r(a)$ to $q$ found by the heavy path routing algorithm is the same as the spanner path, so its length is at most $1 + 2/s + 1/(s - 1)$ times $\abs{pq}$, because that is the length of the spanner path from $r(a)$ to $q$ as proven in Theorem~\ref{thm:spanner}.
\end{proof}

The only thing that remains to be bounded is the length of the path constructed from $p$ to $r(a)$ during the ascending stage.

\begin{lemma}
\label{lem:ascending4}
Let $p = p_1, p_2, \dots, p_k = r(a)$ be the points visited during the ascending stage. For any point $p_i$, the points $p_1$ through $p_{i-1}$ are all stored in the subtree rooted at $h(p_i)$.
\end{lemma}

\begin{proof}
For $i = 1$, the claim is vacuously true. Assume that the claim is true for some $i \ge 1$. The points $p_1, \dots, p_i$ are all stored in the subtree rooted at $h(p_i)$. By definition of the ascending step, $p_{i+1}$ is the representative of the parent of $h(p_i)$ in the compressed quadtree, so $h(p_i)$ is a descendant of $h(p_{i+1})$. Therefore, the points $p_1, \dots, p_i$ are all stored in the subtree rooted at $h(p_{i+1})$. Since $p_{i+1}$ is also stored in that subtree, the proof is complete.
\end{proof}

We can bound the length of the path constructed during the ascending stage using the previous lemma.

\begin{lemma}
\label{lem:ascending-length}
The length of the path constructed from $p$ to $r(a)$ during the ascending stage of the algorithm is no more than $(2/s)\abs{pq}$.
\end{lemma}

\begin{proof}
Recall that $\ell(a)$ is the diagonal length of the hypercube $C_S(a)$ that contains $S(a)$. First, note that if $a$ is the parent of $b$ in the quadtree, then $\ell(a) \ge (1/2)\ell(b)$, by Lemma~\ref{lem:quadtree-cell}. The path from $p$ to $r(a)$ is contained in the subtree rooted at $a$, so the length of every edge on the path is at most $\ell(a)$. That path, minus the last edge, is contained in the subtree rooted at one of the children of $a$ by Lemma~\ref{lem:ascending4}, so the length of all but the last edge is at most $(1/2)\ell(a)$.

Repeating this argument shows that the length of the entire path is not more than $\ell(a) + (1/2)\ell(a) + (1/2)^2\ell(a) + \cdots = 2\ell(a)$. By the condition for checking well-separatedness in the WSPD construction algorithm, $\ell(a) \le (1/s)d(a, b) \le (1/s)\abs{pq}$. Therefore, the length of the path from $p$ to $r(a)$ is at most $(2/s)\abs{pq}$.
\end{proof}

\begin{theorem}
The routing ratio of the heavy path routing algorithm is at most $1 + 4/s + 1/(s - 1)$.
\end{theorem}

\begin{proof}
By Lemma~\ref{lem:ascending-length}, the total length of the path constructed during the ascending stage from $p$ to $r(a)$ is no more than $(2/s)\abs{pq}$. The path constructed during the descending step is equal to the spanner path from $r(a)$ to $q$. The length of the spanner path from $r(a)$ to $q$ is at most $(1 + 2/s + 1/(s-1))\abs{pq}$, as proven in Lemma~\ref{lem:descending-analysis}. Therefore the length of the path from $p$ to $q$ is
\begin{equation*}
    \frac{2}{s}\abs{pq} + \bigg(1 + \frac{2}{s} + \frac{1}{s-1}\bigg)\abs{pq} = \bigg(1 + \frac{4}{s} + \frac{1}{s-1}\bigg)\abs{pq}. \qedhere
\end{equation*}
\end{proof}

Similar to Lemma~\ref{lem:spanner-diameter}, we can bound the number of edges on the spanner path. In fact, the proof is almost identical to the proof of Lemma~\ref{lem:spanner-diameter}.

\begin{lemma}
\label{lem:routing-diameter}
Starting at a point $p$, a message can be forwarded to any other point $q$ after forwarding only $2\lg n + 1$ times.
\end{lemma}

\begin{proof}
Let $\{S(a), S(b)\}$ be the WSPD pair that separates $p$ from $q$. Consider the subpath $p = p_0, p_1, \dots, p_k = r(a)$ found during the ascending stage. There will be one edge added to this path for each light edge on the path from $p$ to $a$ in the compressed quadtree. By Lemma~\ref{lem:light-paths}, this is at most $\lg n$.

Since the path constructed during the descending stage follows the spanner path, Lemma~\ref{lem:spanner-diameter} implies that the number of forwards during the descending stage is at most $\lg n + 1$. Therefore the number of forwards for the entire routing algorithm is at most $2\lg n + 1$.
\end{proof}

The results of this section are summarized in the following theorem.

\begin{theorem}
Let $G$ be a heavy path WSPD spanner for a set $S$ of points in $\R^d$, and let $p$ and $q$ be points of $S$. There exists a local, memoryless routing algorithm that can find a path from $p$ to $q$, such that:
\begin{itemize}
    \item The number of bits stored at each vertex $u$ is $(3\deg(u) + 1)\lg n$
    \item The length of the path found from $p$ to $q$ is at most $(1 + 4/s + 1/(s-1))d(p, q)$
    \item The number of edges on the path is at most $2\lg n + 1$
\end{itemize}
\end{theorem}

Note that the routing ratio is slightly larger than the spanning ratio. This is due to the fact that the path found by Algorithm~\ref{alg:spanner-path} and the path found by the routing algorithm differ in the ascending step. In the routing algorithm, the ascending step ``goes up'' the tree $T$ one light edge at a time. However, in Algorithm~\ref{alg:spanner-path}, the corresponding ascending steps may jump up several levels.

The routing algorithm cannot make these jumps since it has no way of knowing if one of the levels that it skips will be the one that contains the pair $\{S(a), S(b)\}$. In other words, it cannot recognise from a vertex $u$ if one of its neighbours is $r(a)$. It is only when we reach $r(a)$ and have access to its neighbours when we realize that we are at $r(a)$.

In the next section, we give a lower bound that shows our analysis of the routing ratio is fairly tight.

\subsection{A lower bound on the routing ratio}
\label{sec:lower-bound}

We can get a lower bound on the routing ratio by constructing a point set and explicitly computing the routing ratio for a pair of points. In this section we will give an example of a point set where the lower bound is close, but not equal to, the upper bound from the previous section.

Given any $s > 2$, we can construct a set of points in $\R$ that demonstrates this. Fix $k = \ceil{\lg(4s + 8)}$ and let $\alpha = 2^{-k}$. Now let $P = \{\alpha, 3\alpha, 5\alpha, 7\alpha, 1 - 7\alpha, 1 - 5\alpha, 1 - 3\alpha, 1 - \alpha\}$. Construct a compressed quadtree for $P$ with $[0, 1]$ as the hypercube that contains $P$, which is needed for the construction. The compressed quadtree for these points is shown in Figure~\ref{fig:quadtree-example}, with heavy and light edges marked.

Each internal node of the quadtree represents a one-dimensional hypercube, which is just an interval. It is easy to check that $C_S(a)$ and $C_S(b)$ are well-separated. However, for any pair of points in the subtree of $a$, the WSPD pair must be a pair of singletons. For example, we can check that $\{p_2\}$ and $C_S(d)$ are not well-separated since $d(\{p_2\}, C_S(d)) = 3\alpha/2 < 2s\alpha = s \cdot \ell(d)$. So the pair that separates $p_2$ from $p_3$ must be $\{\{p_2\}, \{p_3\}\}$.

The WSPD that results from this quadtree will have $13$ pairs. One of the pairs will be $\{S(a), S(b)\}$. The other $12$ will be pairs of singleton sets, each separating one pair of points in either the subtree of $a$ or the subtree of $b$. The combinatorial structure of the heavy path WSPD spanner for $P$ will be two $4$-cliques joined by a single edge.

Consider routing from $p_4$ to $p_5$. The pair separating these two points is $\{S(a), S(b)\}$. The ascending step forwards the message one light edge up in the compressed quadtree until $r(a) = p_1$ is reached. The first application of the ascending step sends the message to $p_3$. The second application sends it to $p_1$. From there, the descending stage is entered and the message is forwarded to $r(b) = p_8$, and then to $p_5$ since those two points are neighbours. Therefore, the path followed by the routing algorithm is $p_4, p_3, p_1, p_8, p_5$.

Now consider the spanner path from $p_4$ to $p_5$ as computed in Algorithm~\ref{alg:spanner-path}. First, an edge between $r(a) = p_1$ and $r(b) = p_8$ is added. Then an edge is added between the representatives of the pair that separates $p_2$ from $p_3$. Since that pair must be a singleton, an edge is added between those two points. The spanner path is $p_4, p_1, p_8, p_5$, meaning that the routing algorithm forwards the message to one extra vertex along the way.

\begin{figure}
\centering
\includegraphics{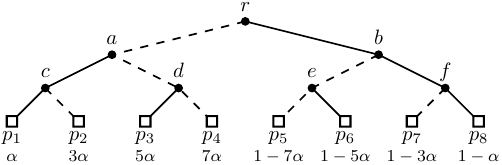}

\vspace{0.5em}

\begin{tabular}{ccc}
\hline
Node $n$ & $C(n)$ & $\ell(n)$ \\
\hline
$r$ & $[0, 1]$ & $1$ \\
$a$ & $[0, 8\alpha]$ & $8\alpha$ \\
$b$ & $[1 - 8\alpha, 1]$ & $8\alpha$ \\
$c$ & $[0, 4\alpha]$ & $4\alpha$ \\
$d$ & $[4\alpha, 8\alpha]$ & $4\alpha$ \\
$e$ & $[1 - 8\alpha, 1 - 4\alpha]$ & $4\alpha$ \\
$f$ & $[1 - 4\alpha, 1]$ & $4\alpha$ \\
\hline
\end{tabular}
\caption[Example point set for a lower bound]{The compressed quadtree for the point set $P$. The heavy edges are solid, and the light edges are dashed. The table gives the quadtree cell for each node, along with the length.}
\label{fig:quadtree-example}
\end{figure}

Since $p_3$ is between $p_4$ and $p_1$, the length of both the routing path and the spanning path will be the same, even though they visit different points. The routing ratio achieved by this example is therefore
\begin{equation}
    \label{eq:routing-ratio-example}
    R = \frac{\abs{p_4p_1} + \abs{p_1p_8} + \abs{p_8p_5}}{\abs{p_4p_5}}.
\end{equation}

The next theorem shows that this fraction can be made arbitrarily close to $1 + 4/s$. This will also imply a lower bound of $1 + 4/s$ on the spanning ratio, since the spanner path has the same length as the routing path in this example. This lower bound is nearly tight, the upper bound for the routing ratio is only $1/(s-1)$ more than the lower bound. The upper bound for the spanning ratio is even closer, the difference between the upper and lower bounds for the spanning ratio is $2/(s-1) - 2/s$.

\begin{theorem}
\label{thm:routing-ratio-lower-bound}
The routing ratio of the heavy path routing algorithm is at least $1 + 4/s$ in the worst case, and the spanning ratio of the heavy path WSPD spanner is at least $1 + 4/s$ in the worst case.
\end{theorem}

\begin{proof}
Let $P$ be the same set of points from the previous example. Notice that we can perturb the points of $P$ slightly without changing the structure of the compressed quadtree or the WSPD. As long as no point is shifted by a distance greater than $\alpha$, the structure will remain the same.

Shifting $p_4$ and $p_5$ inwards and shifting $p_1$ and $p_8$ outwards will simultaneously increase the length of the routing path, and decrease the distance between $p_4$ and $p_5$. This leads to an increased routing ratio.

More precisely, let $0 \le \epsilon < \alpha$ be the amount to shift by. Now change the set $P$ so that $p_1 = \alpha - \epsilon$, $p_4 = 7\alpha + \epsilon$, $p_5 = 1 - 7\alpha - \epsilon$, and $p_8 = 1 - \alpha + \epsilon$.

Now, we can plug values into Equation~\ref{eq:routing-ratio-example}. The routing ratio achieved by this example is
\begin{align*}
R &= \frac{\abs{p_4p_1} + \abs{p_1p_8} + \abs{p_8p_5}}{\abs{p_4p_5}} \\
\begin{split}
=\frac{[(7\alpha + \epsilon) - (\alpha - \epsilon)] + [(1 - \alpha + \epsilon) - (\alpha - \epsilon)]}{[(1 - 7\alpha - \epsilon) - (7\alpha + \epsilon)]} \\
\qquad\qquad\qquad + \frac{[(1 - \alpha + \epsilon) - (1 - 7\alpha - \epsilon)]}{[(1 - 7\alpha - \epsilon) - (7\alpha + \epsilon)]}
\end{split} \\
&= \frac{1 + 10\alpha + 6\epsilon}{1 - 14\alpha - 2\epsilon}.
\end{align*}
The true worst-case routing ratio of the algorithm, $R^*$, is at least $R$ for every value of $\epsilon < \alpha$. In other words, $R^* \ge \sup\{R : 0 \le \epsilon < \alpha\}$. But since $R$ is increasing as a function of $\epsilon$, this can be computed by setting $\epsilon = \alpha$. By doing so we get
\begin{equation*}
    R^* \ge \frac{1 + 16\alpha}{1 - 16\alpha} = 1 + \frac{32\alpha}{1 - 16\alpha}.
\end{equation*}
Finally, due to the way $\alpha$ was chosen we have $2\alpha > (4s + 8)^{-1}$, so
\begin{align*}
    R^* &\ge 1 + \frac{16(4s + 8)^{-1}}{1 - 8(4s + 8)^{-1}} = 1 + \frac{4}{s}
    \qedhere
\end{align*}
\end{proof}

\subsection{Comparing the spanning and routing ratios}
\label{sec:error-analysis}

One more thing we can do is quantify the difference between the spanning and routing ratios. Let $R = 1 + 4/s + 1/(s - 1)$ and $S = 1 + 2/s + 2/(s - 1)$ be the upper bounds on the routing and spanning ratios, respectively. We can see that $R > S$ whenever $s > 2$.

The routing algorithm can be thought of as an approximate shortest path algorithm that only uses local information. With that view we can quantify the difference in terms of absolute and relative error. The absolute error is defined to be $\Delta = R - S = 2/s - 1/(s-1)$ and the relative error is defined to be $\delta = \Delta/S = R/S - 1 = (s-2)/(s^2+3s-2)$. Since $R > S$ for all $s > 2$, both of these errors will be positive.

We can calculate an upper bound for these two quantities. Let $d_R(p, q)$ be the length of the routing path and $d_S(p, q)$ be the length of the spanner path between two given points in a heavy path WSPD spanner. If we have some upper bound $\Delta \le \epsilon$ on the absolute error, then we know that $d_R(p, q) \le d_S(p, q) + \epsilon\abs{pq}$. Likewise, if we have some bound $\delta \le \eta$ on the relative error, then $d_R(p, q) \le (1 + \eta)d_S(p, q)$. The following two lemmas compute actual values for these bounds on the errors.

\begin{lemma}
The absolute error $\Delta(s)$ is bounded by $0.1716$, meaning in any heavy path WSPD spanner with separation ratio $s > 2$ we have $d_R(p, q) \le d_S(p, q) + 0.1716\abs{pq}$ for all points $p$ and $q$.
\end{lemma}

\begin{proof}
To bound $\Delta(s)$ we compute its maximum value. We first differentiate with respect to $s$,
\begin{equation*}
    \Delta'(s) = \frac{1}{(s - 1)^2} - \frac{2}{s^2}.
\end{equation*}
We can then solve the equation $\Delta'(s) = 0$ for $s$ to find that $\Delta$ is maximized at $s^* = 2 + \sqrt{2}$. The maximum value attained is $\Delta(s^*) = 3 - 2\sqrt{2} \approx 0.1716$.
\end{proof}

\begin{lemma}
The relative error $\delta(s) = (s-2)/(s^2+3s-2)$ is bounded by $0.0790$, meaning in any heavy path WSPD spanner we have $d_R(p, q) \le 1.0790\cdot d_S(p, q)$ for all points $p$ and $q$.
\end{lemma}

\begin{proof}
To bound the relative error, we again follow the same procedure. First, compute the derivative of $\delta$ with respect to $s$
\begin{equation*}
    \delta'(s) = -\frac{s^2 - 4s - 4}{s^4 + 6s^3 + 5s^2 - 12s + 4}.
\end{equation*}
We then solve the equation $\delta'(s) = 0$ for $s$ to find that $\delta$ is maximized at $s^* = 2 + 2\sqrt{2}$. The maximum value attained is $\delta(s^*) = (7 - 4\sqrt{2})/17 \approx 0.0790$.
\end{proof}
\section{Doubling spaces}
\label{ch:doubling-spaces}

The local routing algorithm of the previous section relies on very few properties of Euclidean space. The only information that needs to be stored in either the routing tables or the message header are labels of nodes in the compressed quadtree. As such, in this section we extend our results to the class of metric spaces with bounded doubling dimension. In the following section, the local routing algorithm will be generalized to this new setting.

\subsection{Metric spaces}
\label{sec:metric-spaces}

We provide a summary of some geometric notions as described in Har-Peled and Mendel \cite{Har-Peled2006}. Let $X$ be a set. A metric on $X$ is a function that defines the distance between elements in $X$. Formally, a metric is a function $d : X \times X \to \R$ that satisfies the following conditions for all $x$, $y$, and $z$ in $X$:
\begin{itemize}
    \item $d(x, y) \ge 0$, with equality if and only if $x = y$,
    \item $d(x, y) = d(y, x)$; and
    \item $d(x, y) \le d(x, z) + d(z, y)$.
\end{itemize}
The third condition is called the triangle inequality.

Perhaps the prototypical example of a metric space is $\R^d$, with the usual Euclidean metric
\begin{equation*}
    d(x, y) = \norm{y - x} = \sqrt{{\textstyle\sum_{i=1}^d (y_i - x_i)^2}}.
\end{equation*}
Another example is $\R^d$ with the taxicab metric: $d(x, y) = \sum_{i=1}^d \abs{y_i - x_i}$.

We want to consider metric spaces that have similar properties to Euclidean space. One such type of metric space is a \textit{doubling space}. Before we define a doubling space, we need a few basic definitions.

A (closed) ball in a metric space $X$ is the set of all points whose distance to a given point $x \in X$ is at most $r$. The point $x$ is called the centre of the ball and $r$ is called the radius. The closed ball centred at $x$ with radius $r$, written $B(x, r)$, is the set $\{y \in X : d(x, y) \le r\}$. Replacing the inequality with a strict one in the definition gives an open ball. 

Let $A$ be a subset of $X$. We call $A$ an $\epsilon$-cover of $X$ if for any point $x$ in $X$, there is some point $y$ in $A$ with $d(x, y) \le \epsilon$. In other words, the union of all closed balls with radius $\epsilon$ centred at the points of $A$ contains $X$.

Let $B$ be a subset of $X$. We call $B$ an $\epsilon$-packing of $X$ if the distance between any two points of $B$ is at least $\epsilon$. Sometimes, $B$ is called $\epsilon$-discernible. This is equivalent to the condition that the balls of radius $\epsilon/2$ centred at the points of $B$ are pairwise disjoint.

A set $C$ is called an $\epsilon$-net of $X$ if it is both an $\epsilon$-cover and an $\epsilon$-packing of $X$. See Figure~\ref{fig:epsilon-net} for an illustration.

\begin{figure}
\centering
\includegraphics{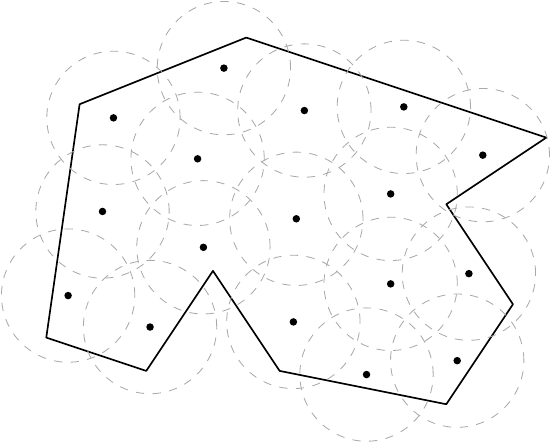}
\caption[An $\epsilon$-net]{An $\epsilon$-net of a simple polygon. The points of the net are drawn in black. A disk of radius $\epsilon$ has been drawn around each black point using dashed grey lines. Every point in the polygon lies in at least one of the disks, and none of the black points lie in a disk centred at another black point.}
\label{fig:epsilon-net}
\end{figure}

A metric space is called doubling if there exists a constant $K \ge 0$ such that any ball of radius $r$ can be covered by $K$ balls of radius $r/2$. Let $K^*$ be the smallest such constant. We call $K^*$ the doubling constant of $X$, and $\lg K^*$ the doubling dimension of $X$.

A metric space with bounded doubling dimension can also be called a doubling space. We will use these two terms interchangeably throughout the rest of the article.

\subsection{Heavy path WSPD spanners in doubling spaces}
\label{sec:hpw-spanners-doubling}

The construction of a heavy path WSPD spanner can be generalized to doubling spaces. The procedure remains largely the same, but a few key changes are necessary. First we need a replacement for the compressed quadtree, because there is no general way to define a hypercube in a doubling space. We will use the \emph{net tree} \cite{Har-Peled2006}. From there, the construction of a heavy path WSPD spanner is largely unchanged.

Our definition of well-separated used the notion of Euclidean distance. We simply replace that with any metric to get a more general definition of well-separated. Explicitly, let $S$ and $T$ be two point sets in a metric space $(X, d)$. We say that $S$ and $T$ are well-separated with respect to $s > 0$ if $d(S, T) \ge s \cdot \max\{\diam S, \diam T\}$, where $d(S, T) = \min\{d(p, q) : p \in S, q \in T\}$ and $\diam S$ is $\max\{d(x,y) : x, y \in X\}$. The number $s$ is called the separation ratio.

We can also generalize Lemma~\ref{lem:ws-pairs} to general metric spaces by simply replacing the Euclidean distance with any metric. The proof is identical to that of Lemma~\ref{lem:ws-pairs}, so we can simply restate the result here.

\begin{lemma}
\label{lem:general-ws-pairs}
Let $S$ and $T$ be well-separated point sets in some metric space $(X, d)$ with respect to $s > 0$. Then for any points $p, p' \in S$ and $q, q' \in T$
\begin{enumerate*}[label={(\alph*)},before=\unskip{: },itemjoin={{; }},itemjoin*={{; and }},after={.}]
    \item $d(p, p') \le (1/s)d(p, q)$
    \item $d(p', q') \le (1 + 2/s)d(p, q)$
\end{enumerate*}
\end{lemma}

We cannot compute a compressed quadtree for a set of points in an arbitrary metric space, so we will require some other data structure that has similar properties. In \cite{Har-Peled2006}, Har-Peled and Mendel describe how to construct a tree that they call a \emph{net tree} that has all the desired properties. They also show how to construct a WSPD and a WSPD spanner from that tree. The routing algorithm will be generalized to work on that spanner.

Let $S$ be a set of $n$ points in a metric space of doubling dimension $\lambda$. Like the compressed quadtree, the net tree is a data structure that stores the points of $S$ in its leaves. An internal node $a$ corresponds to the set $S(a)$ of points stored in the leaves of the subtree of $a$. The net tree also has similar properties to the compressed quadtree that will make constructing a WSPD of linear size and local routing possible. 

In particular, the net tree has the following properties as outlined in \cite{Har-Peled2006}. Let $\tau \ge 11$ be a constant chosen before constructing the net tree. Every non-leaf node $a$ in $T$ has a level $\level(a) \in \Z \cup \{-\infty\}$ with the property that the level of a node is less than the level of its parent. The level of a leaf is defined to be $-\infty$. Additionally, every node has a representative $r(a)$ with the property that if $p$ is the representative of an internal node $a$, then $a$ has some child $b$ for which $r(a) = r(b)$. The parent of a node $a$ is denoted $p(a)$, and the set of points stored in the subtree rooted at $a$ is denoted $S(a).$

There are two more properties of the net tree that make it appropriate for our application:
\begin{enumerate}
    \item{} [\textsl{Covering property}] For every node $a$ of $T$,
    \begin{equation*}
        \label{eq:net-tree-covering}
        S(a) \subset B\Big(r(a), \frac{2\tau}{\tau-1}\tau^{\level(a)}\Big).
    \end{equation*}
    \item{} [\textsl{Packing property}] For every non-root node $a$ of $T$,
    \begin{equation*}
        \label{eq:net-tree-packing}
        B\Big(r(a), \frac{\tau - 5}{2(\tau-1)}\tau^{\level(p(a)) - 1}\Big) \subset S(a).
    \end{equation*}
\end{enumerate}

Importantly, the covering property implies that, for every node $a$ in a net tree, the diameter of $S(a)$ is at most $\frac{4\tau}{\tau-1}\tau^{\level(a)}$.

The details of the construction of a net tree can be found in Har-Peled and Mendel \cite{Har-Peled2006}. We summarize their bound on the running time in the following theorem.

\begin{theorem}[{\autocite[Theorem~3.1]{Har-Peled2006}}]
Given a set of $n$ points in a metric space of doubling dimension $\lambda$, a net tree can be computed in $O(2^\lambda n \log n)$ expected time.
\end{theorem}

The following lemma is analogous to Lemma~\ref{lem:quadtree-cell}, where the diameter of a compressed quadtree node was bounded by a constant fraction of the diameter of its parent. This lemma will be used in a similar way, to bound the length of the ascending stage of the routing algorithm.

\begin{lemma}
\label{lem:nettreediameter}
Let $T$ be a net-tree, and let $a$ be any node of the net tree. Then we have
\begin{equation}
\label{eq:nettreediameter}
    \diam{S(a)} \le \frac{4\tau}{\tau - 1}\tau^{\ell(p^k(a)) - k}
\end{equation}
where $p^k(a)$ is the $k$-th ancestor of $a$: $p^0(a) = a$ and $p^k(a) = p(p^{k-1}(a))$ for $k > 0$.
\end{lemma}

\begin{proof}
The proof is by induction on $k$. The base case follows directly from the covering property of net trees. Since $S(a)$ can be contained in a ball of radius $r = \frac{2\tau}{\tau - 1}\tau^{\ell(a)}$, its diameter is at most $2r$.

Now assume that Equation~(\ref{eq:nettreediameter}) holds for some $k \ge 0$. Recall that for every non-root node $a$ in $T$, we have $\ell(a) < \ell(p(a))$. Since the level of a node is an integer this is equivalent to $\ell(a) \le \ell(p(a)) - 1$. Applying the induction hypothesis and then this inequality to the node $p^k(a)$ yields
\begin{align*}
    \diam{S(a)}
    &\le 2\frac{2\tau}{\tau - 1}\tau^{\ell(p^k(a)) - k} \\
    &\le 2\frac{2\tau}{\tau - 1}\tau^{\ell(p^{k+1}(a)) - (k+1)}. \qedhere
\end{align*}
\end{proof}

Given a net tree $T$, the procedure for computing a WSPD is similar to the Euclidean case. The following algorithm will produce a WSPD for $P$ with separation $s$, when given the root of $T$ as both inputs. The algorithm is from Har-Peled and Mendel \cite{Har-Peled2006}.

\begin{algorithm}
\caption{Computing a WSPD for a set of points in a doubling space}
\label{alg:wspd-doubling}
\begin{algorithmic}
\State \textbf{Input}: $a$ and $b$ are nodes of a net tree $T$ that stores a set $S$ of points in a metric space of bounded doubling dimension, and $s > 2$ is the separation ratio
\State \textbf{Output}: if initially called with both $a$ and $b$ equal to the root of $T$, the algorithm outputs a WSPD of $S$ with separation ratio $s$
\Procedure{WSPD}{$a, b$}
\State \algorithmicif\ $a = b = \{p\}$ for some point $p$\ \algorithmicthen\ \Return $\emptyset$
\If{$\level(a) < \level(b)$}
    \State swap $a$ and $b$
    \Comment{now $\level(a) \le \level(b)$}
\EndIf
\If{$8s\frac{2\tau}{\tau - 1}\cdot\max\{\tau^{\level(a)}, \tau^{\level(b)}\} \le d(r(a), r(b))$}
    \State \Return $\{\{S(a), S(b)\}\}$
     \Comment{$S(a)$ and $S(b)$ are well-separated}
\Else
    \State let $a_1, a_2, \dots, a_k$ be the children of $a$
    \State \Return $\bigcup_i \Call{WSPD}{a_i, b}$
\EndIf
\EndProcedure
\end{algorithmic}
\end{algorithm}

The only difference between this algorithm and the one to compute a WSPD given a quadtree is the condition for checking if two sets are well-separated.

\begin{theorem}[{\plaincite[Lemma~5.1]{Har-Peled2006}}]
Given a set of $n$ points in a metric space of doubling dimension $\lambda$, a WSPD with separation $s > 2$ can be computed in $O(2^\lambda n\log n + s^\lambda n)$ expected time. The number of pairs is $O(s^\lambda n)$.
\end{theorem}

Notice that the nodes in a net tree have representatives, by definition. The representatives of the net tree satisfy Lemma~\ref{lem:representative-hierarchy}, so we could use them in the construction of the spanner. However, they do not necessarily satisfy Lemma~\ref{lem:light-paths}. This means that if a spanner is constructed using these representatives, the hop spanning ratio cannot be bounded like in Lemma~\ref{lem:spanner-diameter}. Instead, we can choose the representatives according to a heavy path decomposition, as we did in Euclidean space, since the representatives computed in the construction of the net tree are arbitrary. Since a heavy path decomposition can be computed in $O(n)$ time, this does not increase the running time of the algorithm. From this point on, the representative $r(a)$ of a node will be defined by the heavy path decomposition.

Finally, the construction of a spanner from a WSPD is identical. For each pair $\{S(a), S(b)\}$ in the WSPD, we add an edge between $r(a)$ and $r(b)$. Again, $r(a)$ and $r(b)$ are chosen according to the heavy path decomposition, they are not the representatives from the definition of the net tree.

The proof that this graph was a spanner in Euclidean space also relied on Lemma~\ref{lem:ws-pairs}, which can be verified to hold for any WSPD in any metric space, and Lemma~\ref{lem:wspd-subtree}. Notice that the WSPD computed in this section satisfies the hypothesis for Lemma~\ref{lem:wspd-subtree}, namely that every pair in the WSPD is of the form $\{S(a), S(b)\}$ for some nodes $a$, $b$ in the net tree. So everything that was used to prove that the WSPD graph was a spanner in Section~\ref{sec:hpw-spanners} has been shown to hold in a metric space of bounded doubling dimension when we use a net tree. We have the following theorem, which is analogous to Theorem~\ref{thm:spanner}.

\begin{theorem}
\label{thm:doubling-spanner}
Given a set $P$ of $n$ points in a metric space with doubling dimension $\lambda$, a $(1 + 2/s + 2/(s - 1))$-spanner of $P$ with $O(s^\lambda n)$ edges can be computed in $O(2^\lambda n \log n + s^\lambda n)$ expected time.
\end{theorem}

Finally, we mention that Lemma~\ref{lem:spanner-diameter} holds for the net tree as well, since we chose the representatives using a heavy path decomposition. Therefore, Lemma~\ref{lem:light-paths} can be applied to a net tree as well, and the proof of Lemma~\ref{lem:spanner-diameter} only relied on that lemma. The construction of a heavy path WSPD spanner in a metric space of bounded doubling dimension is summarized by the following theorem.

\begin{theorem}
\label{thm:hpwsd-summary}
Let $S$ be a set of $n$ points in a metric space $X$ of doubling dimension $\lambda$, and let $s > 2$. In $O(2^\lambda n\log n + s^\lambda n)$ expected time, we can construct a graph $G$ called a heavy path WSPD spanner with the following properties:
\begin{itemize}
    \item The number of edges in $G$ is $O(s^\lambda n)$.
    \item $G$ is a $(1 + 2/s + 2/(s - 1))$-spanner.
    \item $G$ is a $(2\lg n + 1)$-hop spanner.
\end{itemize}
Additionally, between any two points there is a single path that achieves both the spanning and hop-spanning ratio.
\end{theorem}

\section{Local routing in doubling spaces}
\label{ch:doubling-routing}

We wish to apply our routing algorithm to a heavy path WSPD spanner constructed in some metric space $(X, d)$ with bounded doubling dimension. We have already seen that the compressed quadtree cannot be used, and so we have substituted the net tree. From there, the construction of a spanner is very similar to the Euclidean case. Now the question is what, if anything, has to be modified in order to route in this more general setting.

The labelling scheme, as stated, can be applied to any rooted tree. So we can label the points in the same manner, using $\lg n$ bits per point, and $2\lg n$ bits for each internal node.
 
None of the arguments in the proof of correctness for the routing algorithm made reference to the properties of the compressed quadtree or Euclidean space. Only properties of the WSPD and the heavy path decomposition were needed. We have already shown that the WSPD computed using the net tree satisfies all the relevant properties. In particular, the proofs of Section~\ref{sec:routing-algorithm} rely on Lemmas~\ref{lem:wspd-subtree} and \ref{lem:representative-hierarchy}, which also apply to the WSPD constructed in a doubling space.

Since we computed a heavy path decomposition of the net tree, we know that Lemma~\ref{lem:spanner-diameter} holds. The proof of Lemma~\ref{lem:routing-diameter} only relied on the use of a heavy path decomposition, and so it will hold here as well.

\subsection{Analysis for doubling spaces}
\label{sec:doubling-analysis}

Part of the correctness proof showed that the descending stage followed the spanner path from $r(a)$ to $q$. Since the length of the spanner path is the same for both Euclidean and doubling spaces (see Theorem~\ref{thm:doubling-spanner}), the length of the path constructed during the descending stage is $(1 + 2/s + 1/(s-1))d(p, q)$, as it was in Euclidean space. The analysis of the ascending stage, however, relied on properties of the compressed quadtree to bound the length of the path. The following lemma bounds the length of the ascending stage using Lemma~\ref{lem:nettreediameter} instead.

\begin{lemma}
\label{lem:doubling-ascending-length}
The length of the path constructed during the ascending step of the algorithm is no more than
\begin{equation*}
    \frac{\tau}{s(\tau-1)}d(p, q).
\end{equation*}
\end{lemma}

\begin{proof}
The path from $p$ to $r(a)$ is contained in the subtree rooted at $a$, so the length of every edge on the path is at most $\diam(S(a)) \le 2\frac{2\tau}{\tau - 1} \tau^{\level(a)}$. That path, minus the last edge (i.e., the edge with endpoint $r(a)$), is contained in the subtree rooted at one of the children of $a$ by Lemma~\ref{lem:ascending4}, so the length of all but the last edge is at most $2\frac{2\tau}{\tau - 1} \tau^{\level(a)-1}$, by Lemma~\ref{lem:nettreediameter}.

Repeating this argument will show that the length of the entire path is not more than
\begin{align*}
    \sum_{k=0}^\infty 2\frac{2\tau}{\tau - 1} \tau^{\level(a)-k}
    &= 2\frac{2\tau}{\tau - 1}\tau^{\level(a)} \cdot \sum_{k=0}^\infty \tau^{-k} \\
    &= 2\frac{2\tau}{\tau - 1}\tau^{\level(a)} \cdot \frac{\tau}{\tau-1}.
\end{align*}
From Algorithm~\ref{alg:wspd-doubling} we know that
\begin{align*}
    2\frac{2\tau}{\tau - 1} \tau^{\level(a)}
    &\le \frac{1}{4s}d(r(a), r(b)) \\
    &\le \frac{1}{4s}\bigg(1+\frac{2}{s}\bigg)d(p, q) \\
    &\le \frac{1}{s}d(p, q),
\end{align*}
where the last inequality follows because $s \ge 2$. The first inequality is the condition for being well-separated in Algorithm~\ref{alg:wspd-doubling}. And so therefore the length of the path is at most
\begin{equation*}
    2\frac{2\tau}{\tau - 1}\tau^{\level(a)} \cdot \frac{\tau}{\tau-1}
    \le \frac{\tau}{s(\tau-1)}d(p, q).
    \qedhere
\end{equation*}
\end{proof}

Now we have all the ingredients to analyze the routing algorithm in doubling spaces. The following theorem summarizes the result.

\begin{theorem}
The routing ratio of this algorithm on a heavy path WSPD spanner in a doubling space is at most
\begin{equation*}
    1 + \bigg(2 + \frac{\tau}{\tau - 1}\bigg)\frac{1}{s} + \frac{1}{s - 1}.
\end{equation*}
\end{theorem}

\begin{proof}
Since the descending stage follows the spanner path, its length is at most $(1 + 2/s + 1/(s - 1))d(p, q)$, just like in the Euclidean case. By Lemma~\ref{lem:doubling-ascending-length}, the length of the ascending stage is at most $\tau/(\tau - 1) \cdot (1/s) \cdot d(p, q)$. Since the ascending and descending stages make up the entire algorithm, we can add these two bounds to get the routing ratio.
\end{proof}

Because we know that $\tau \ge 11$, this theorem implies that the routing ratio in doubling spaces is at most $1 + 3.1/s + 1/(s-1)$. This is smaller than the routing ratio in the Euclidean case. However, the constants hidden by the big-O notation in the net tree construction will depend on $\tau$, and so the number of edges in the heavy path WSPD spanner will also depend on the choice of $\tau$.

As we did in the Euclidean case, we can analyze the difference between the routing and spanning ratios. Let $S = 1 + 2/s + 2/(s - 1)$ and $R = 1 + (2 + \tau/(\tau-1))/s + 1/(s - 1)$ be the spanning and routing ratios. Before, the spanning ratio was always less than the spanning ratio. Here, however, the spanning ratio is only less than the routing ratio if $s > \tau$. Otherwise, the inequality is flipped. Since the routing algorithm produces a path on the same graph, this implies a new upper bound for the spanning ratio.

The length of the path constructed during the ascending stage depends on the constant $\tau$. Since we are free to choose any $\tau \ge 11$ when constructing the net tree, if we choose $\tau$ to be large enough then the length of the path will approach $(1/s)d(p, q)$, which is less than the length of the spanner path from $p$ to $r(a)$. Using this fact, we can get a better bound on the spanning ratio of the heavy path WSPD spanner.

\begin{theorem}
If $s \le \tau$, then the spanning ratio of the graph of Theorem~\ref{thm:doubling-spanner} is at most $1 + 4.2/s$.
\end{theorem}

\begin{proof}
Let $p, q$ be points of $P$, and let $\{S(a), S(b)\}$ be the WSPD pair that separates $p$ from $q$. Consider routing from $p$ to $q$. The ascending stage constructs a path from $p$ to $r(a)$ with length at most
\begin{equation*}
    \frac{\tau}{s(\tau-1)}d(p, q).
\end{equation*}
If we run the routing algorithm from $q$ to $p$ instead, we get a path from $q$ to $r(b)$ with the same length.

The distance between $r(a)$ and $r(b)$ is at most $(1 + 2/s)d(p, q)$. Add the two paths constructed by the ascending stages to this edge to get a path of length at most
\begin{align*}
    d_G(p, q)
    &\le d_G(p, r(a)) + d_G(r(a), r(b)) + d_G(r(b), q) \\
    &\le
    \frac{\tau}{s(\tau-1)}d(p, q) + \Big(1 + \frac{2}{s}\Big)d(p, q) + \frac{\tau}{s(\tau-1)}d(p, q)\\
    &= \Big(1 + \Big(2 + \frac{2\tau}{\tau-1}\Big)\frac{1}{s}\Big)d(p, q).
\end{align*}
Since $\tau \ge 11$, we have $\tau/(\tau-1) \le 1.1$, and the inequality follows.
\end{proof}

We can also bound the number of edges on the routing path, like Lemma~\ref{lem:routing-diameter}. The proof of that lemma relied on Lemma~\ref{lem:light-paths} and Lemma~\ref{lem:spanner-diameter}, which we have show to hold in doubling spaces. Therefore the number of edges on the routing path has the same upper bound of $2 \lg n + 1$. The results of this section are summarized in the following theorem.

\begin{theorem}
Let $G$ be a heavy path WSPD spanner for a set $S$ of points in a metric space of doubling dimension $\lambda$, and let $p$ and $q$ be points of $S$. There exists a local, memoryless routing algorithm that can find a path from $p$ to $q$, such that:
\begin{itemize}
    \item The number of bits stored at each vertex $u$ is $(3\deg(u) + 1)\lg n$
    \item The length of the path found from $p$ to $q$ is at most
    \begin{equation*}
        \bigg(1 + \bigg(2 + \frac{\tau}{\tau - 1}\bigg)\frac{1}{s} + \frac{1}{s - 1}\bigg)d(p, q),
    \end{equation*}
    where $\tau \ge 11$ is a constant
    \item The number of edges on the path is at most $2\lg n + 1$
\end{itemize}
\end{theorem}

\subsection{Error analysis redux}
\label{sec:doubling-error}

As before, we have upper bounds on the spanning ratio and the routing ratio. We can again quantify the difference between them using absolute and relative error. But since we have two different bounds on the spanning ratio depending on whether or not $s \le \tau$, we need to consider two cases.

Recall $d_R(p, q)$ is the length of the routing path and $d_S(p, q)$ is the length of the spanner path between two points, $p$ and $q$. If $R$ and $S$ are our upper bounds on the routing and spanning ratios, then the absolute error is $\Delta = R - S$ and the relative error is $\delta = R/S - 1$.

Our two bounds on the spanning ratio are
\begin{equation*}
    S_{\le\tau} = 1 + \bigg(2 + \frac{2\tau}{\tau-1}\bigg)\frac{1}{s}
    \quad\text{and}\quad
    S_{>\tau} = 1 + \frac{2}{s} + \frac{2}{s-1}.
\end{equation*}
In both cases, the routing ratio is the same
\begin{equation*}
    R = 1 + \bigg(2 + \frac{\tau}{\tau-1}\bigg)\frac{1}{s} + \frac{1}{s - 1}.
\end{equation*}
Notice that if $s = \tau$, then all three of these values coincide. If $s < \tau$, then we have $S_{\le\tau} < R < S_{>\tau}$. And if $s > \tau$, then we have $S_{>\tau} < R < S_{\le\tau}$.

\begin{lemma}
The absolute error $\Delta$ is bounded by $0.5$, meaning in any heavy path WSPD spanner we have $d_R(p, q) \le d_S(p, q) + 0.5d(p, q)$ for all points $p$ and $q$, and for any choice of $s$ and $\tau$.
\end{lemma}

\begin{proof}
We will consider two separate cases: $s \le \tau$ and $s > \tau$. In the first case, the absolute error is
\begin{equation*}
    \Delta_{\le\tau} = R - S_{\le\tau} = \frac{1}{s-1} - \frac{\tau}{s(\tau - 1)}.
\end{equation*}
In the second case, it is
\begin{equation*}
    \Delta_{>\tau} = R - S_{>\tau} = \frac{\tau}{s(\tau-1)} - \frac{1}{s-1}.
\end{equation*}
Interestingly, $\Delta_{\le\tau} = -\Delta_{>\tau}$.

It can be shown that $\Delta_{\le\tau}$ is decreasing on the interval $[2, \tau]$, so its maximum is achieved at $s = 2$. For $\Delta_{>\tau}$, we find that the function achieves a maximum at $s = \tau + \sqrt{\tau^2 - \tau}$. Computing the spanning ratios for these values of $s$, we find
\begin{align*}
    \Delta_{\le\tau}(2) &= \frac{2 - \tau}{2 - 2\tau} \\
    \Delta_{>\tau}(\tau + \sqrt{\tau^2 - \tau}) &= \frac{2\tau}{(\tau-1)(\tau + \sqrt{\tau^2 - \tau})} - \frac{2}{(\tau + \sqrt{\tau^2 - \tau})-1}.
\end{align*}
Finally, we can check that $\Delta_{\le\tau}(2) > \Delta_{>\tau}(\tau + \sqrt{\tau^2 - \tau})$ for all $\tau \ge 11$, that $\Delta_{\le\tau}(2)$ is increasing as a function of $\tau$, and that $\lim_{\tau\to\infty}\Delta_{\le\tau}(2) = 1/2$.
\end{proof}

\begin{lemma}
The absolute error $\delta$ is bounded by $0.1667$, meaning in any heavy path WSPD spanner we have $d_R(p, q) \le 1.1667d(p, q)$ for all points $p$ and $q$, and for any choice of $s$ and $\tau$.
\end{lemma}

\begin{proof}
Again the analysis is split into two separate cases: $s \le \tau$ and $s > \tau$. In the first case, the relative error is
\begin{equation*}
    \delta_{\le\tau} = \frac{s - \tau}{(s^2 + 3s - 4)(1 - \tau)+2s+2}.
\end{equation*}
In the second case, it is
\begin{equation*}
    \delta_{>\tau} = \frac{2(s-\tau)}{(s^2+3s-2)(\tau-1)}.
\end{equation*}

Once again, $\delta_{\le\tau}$ is decreasing on the interval $[2, \tau]$ and so is maximized when $s = 2$. The function $\delta_{>\tau}$ is maximized at the point $s = \tau + \sqrt{\tau^2 + 3\tau - 2}$. We see that
\begin{equation*}
    \delta_{\le\tau}(2) = \frac{\tau - 2}{6\tau - 4}.
\end{equation*}
It can be shown that $\delta_{\le\tau}(2)$ is increasing as a function of $\tau$ and that the limit as $\tau$ approaches infinity is $1/6$.
It can also be shown that $\delta_{>\tau} < 1/6$ for any values of $s$ and $\tau$, therefore the maximum relative error is $1/6$.
\end{proof}

\section{Conclusion}
\label{ch:conclusion}

The main contribution of this article is a competitive local routing algorithm for heavy path WSPD spanners that succeeds for points in a metric space of bounded doubling dimension, with routing ratio $1 + O(1/s)$.

Given a WSPD with separation ratio $s > 2$ for a set of $n$ points, we showed how to construct a spanner with spanning ratio $1 + 2/s + 2/(s - 1)$ and hop spanning ratio $2\lg n + 1$, which we called a heavy path WSPD spanner. In Euclidean space the construction takes $O(d(n \log n + s^dn))$ time and in a metric space with doubling dimension $\lambda$ the construction takes $O(2^\lambda n\log n + s^\lambda n)$ expected time. The number of edges in the spanner is $O(s^dn)$ in $\R^d$, or $O(s^\lambda n)$ in a doubling space.

We then presented a local memoryless routing algorithm for heavy path WSPD spanners. We showed that, if the spanner was constructed from a compressed quadtree in Euclidean space, the routing ratio of the algorithm is at most $1 + 4/s + 1/(s - 1)$. We also provided a worst-case lower bound of $1 + 4/s$. The number of edges on the path generated by the routing algorithm is at most $2\lg n + 1$.

We also analyzed the routing ratio for point sets in metric spaces of bounded doubling dimension. If the WSPD was constructed using the net tree of Har-Peled and Mendel, then the routing ratio is at most $1 + (2 + \tau/(\tau - 1))/s + 1/(s - 1)$, where $\tau \ge 11$ is constant. The bound on the number of edges on the path remains $2\lg n + 1$. Notice that the analysis of the routing ratio depends on the properties of the metric space and the tree used to construct the WSPD. Finally we gave a better bound on the spanning ratio of a heavy path WSPD spanner in a doubling space in the special case where $s \le \tau$.

The spanning ratio of a heavy path WSPD spanner and the different routing ratios are all $1 + O(1/s)$, but the constants differ. Ideally the routing ratio would be the same as the spanning ratio, but this may not be possible. We analyzed the difference between the bounds in an effort to quantify how far from the ideal we are.

\subsection{Future work}
\label{sec:future-work}

Given that the routing algorithm generalized so naturally to doubling spaces, it is natural to wonder if it could generalize to WSPD spanners constructed in other types of metric spaces. Is there a more general class of metric spaces that supports competitive local routing, using the heavy path WSPD spanner?

Another avenue to explore is considering other spanner constructions. The spanner construction that we presented is relatively simple. There are more elaborate constructions that can produce spanners with desirable properties. For example, there are constructions of WSPD spanners with bounded degree. If we could apply our algorithm to these spanners then the size of the routing tables at each node would be $O(\log n)$.

There is also a construction of a WSPD spanner with constant hop spanning ratio. We showed that a heavy path WSPD spanner has a hop spanning ratio of at most $2\lg n + 1$. In some applications, the number of edges on the routing path might be more important than the total length of the path. For example, if the time for a message to be transmitted is much smaller than the time needed to make a forwarding decision. In this case routing on spanners with a small hop spanning ratio would be desirable. For details on the various constructions of spanners from WSPDs, see Narasimhan and Smid \cite{Narasimhan2007}.

Finally, WSPDs have been defined for the unit disk graph \cite{Gao2006}. Using these WSPDs, local routing in the unit disk graph is possible \cite{Kaplan2018}. Our algorithm routes on spanners constructed directly from a WSPD, so the result does not immediately transfer. However, we do not need to use a modifiable header, and it would be interesting to see if our algorithm could be modified to work on the unit disk graph.

\printbibliography

\end{document}